\documentclass[11pt,twoside]{article}
\usepackage{graphicx}
\usepackage{psfrag}
\usepackage{amsmath,amssymb,amsthm,mathrsfs,amsfonts,amscd}
\usepackage{cite}
\usepackage{color}
\usepackage{hyperref}
\usepackage{cleveref}
\usepackage{authblk}
\usepackage{appendix}

\newtheorem{theorem}{Theorem}[section]
\newtheorem{lemma}[theorem]{Lemma}
\newtheorem{corollary}[theorem]{Corollary}
\newtheorem{prop}[theorem]{Proposition}
\newtheorem{remark}[theorem]{Remark}

\renewcommand{\P}{\mathbb{P}}

\renewcommand{\S}{\mathbb{S}}
\newcommand{\R}{\mathbb{R}}
\newcommand{\E}{\mathbb{E}}

\newcommand{\lp}{\left(}
\newcommand{\rp}{\right)}
\newcommand\inp[2]{\left\langle #1, #2 \right\rangle}
\newcommand\step[1]{
\vskip .1in
{\bf \noindent #1}
\vskip .1in
}

\providecommand{\keywords}[1]
{
  \textbf{\textbf{Keywords:}} #1
}

\usepackage{fancyhdr}
\title{\bf Sub-Gaussian Matrices on Sets: \\ Optimal Tail Dependence and Applications}

\author{Halyun Jeong}
\author{Xiaowei Li}
\author{Yaniv Plan}
\author{{\"O}zg{\"u}r Y{\i}lmaz}
\affil{Department of Mathematics, The University of British Columbia\\
\vspace{.1cm}\texttt{\{hajeong,xli,yaniv,oyilmaz\}@math.ubc.ca}}

\begin{document}
\maketitle

\begin{abstract}
Random linear mappings are widely used in modern signal processing, compressed sensing and machine learning. These mappings may be used to embed the data into a significantly lower dimension while at the same time preserving useful information. This is done by approximately preserving the distances between data points, which are assumed to belong to $\R^n$.  Thus, the performance of these mappings is usually captured by how close they are to an isometry on the data. Gaussian linear mappings have been the object of much study, while the sub-Gaussian settings is not yet fully understood. In the latter case, the performance depends on the sub-Gaussian norm of the rows. In many applications, e.g., compressed sensing, this norm may be large, or even growing with dimension, and thus it is important to characterize this dependence. 

We study when a sub-Gaussian matrix can become a near isometry on a set, show that previous best known dependence on the sub-Gaussian norm was sub-optimal, and present the optimal dependence. Our result not only answers a remaining question posed by Liaw, Mehrabian, Plan and Vershynin in 2017, but also generalizes their work.
We also develop a new Bernstein type inequality for sub-exponential random variables, and a new Hanson-Wright inequality for quadratic forms of sub-Gaussian random variables, in both cases improving the bounds in the sub-Gaussian regime under moment constraints.
Finally, we illustrate popular applications such as Johnson-Lindenstrauss embeddings, null space property for 0-1 matrices, randomized sketches and blind demodulation, whose theoretical guarantees can be improved by our results (in the sub-Gaussian case). 

\end{abstract}
 

\section{Introduction}
Random linear mappings play a central role in dimension reduction, compressed sensing, and numerical linear algebra due to their propensity to preserve the geometry of a given set. The performance of a random linear mapping $A\in\R^{m\times n}$ is often determined by the uniform concentration bound of $\frac{1}{\sqrt{m}}\|Ax\|_2$ around $\|x\|_2$ for all vectors in a set of interest (in other words, how close the map $\frac{1}{\sqrt{m}}A$ is to being an isometry on the set).
This is now well-understood by the standard techniques in the Gaussian random matrix case \cite{Schechtman2006, vershynin_2018,gordon1988milman}. 
However, in many applications, non-Gaussian random mappings are more useful because of their computational/storage benefits or simply the difficulty to generate Gaussian matrices using sampling devices \cite{krahmer2015compressive}. For example, sparse or structured random matrices are preferred in both dimension reduction \cite{Dirksen2016dimension} and random sketching in numerical linear algebra \cite{achlioptas2003database, kane2014sparser, 2015pilanci, 2014woodruff} since they provide more efficient matrix multiplications than dense and unstructured matrices such as Gaussian ones. Certain formulations in compressed sensing also naturally require random matrices such as randomly subsampled Fourier measurements \cite{Krahmer2014} or Bernoulli random matrices \cite{saab2018compressed}.

There has been a series of recent works \cite{Dirksen2016dimension, liaw2017simple,2017oymak} to demonstrate the effectiveness of random mappings outside the Gaussian setup. Unlike the Gaussian case in which we have a rotation invariance property, non-Gaussian setups require more sophisticated arguments to address various new technical challenges. In this article, we will be focusing on sub-Gaussian random mappings.

Let us recall some definitions. For $\alpha \geq 1$, the $\psi_\alpha$-norm (which is the Orlicz norm taken with respect to function $\exp(x^\alpha) -1$) of a random variable $X$ is defined as 
$$ \| X\|_{\psi_\alpha} := \inf\{ t>0: \E \exp(|X|^\alpha/t^\alpha) \leq 2\}. $$
In particular, $\alpha=2$ gives the sub-Gaussian norm and $\alpha=1$ gives the sub-exponential norm. The random variable $X$ is called sub-Gaussian if $\| X\|_{\psi_2}<\infty$ and called sub-exponential if $\| X\|_{\psi_1}<\infty$.

For sub-Gaussian random variables, the $\psi_2$-norm roughly measures how fast the tail distribution decays -- usually the bigger $\psi_2$-norm is, the heavier the tail. We will repeatedly use the fact that $\|X\|_{\psi_2}\leq K$ if and only if the tail probability $\P(|X|\geq t)$ is bounded by a Gaussian with standard deviation in the order of $K$. A precise statement of this, along with some other properties of $\psi_\alpha$-norm, can be found in \Cref{appendix_psi_alpha_properties}.

The sub-Gaussian norms for many random variables can be calculated by looking at the moment generating function of their squares. For example, 
the sub-Gaussian norm for $\mathbf{Normal}(0,\sigma^2)$ is $\sqrt{\frac{8}{3}}\,\sigma$; for $\mathbf{Bernoulli}(p)$ it is $\log^{-\frac{1}{2}}\lp 1+p^{-1} \rp$; for Rademacher random variable it is $\log^{-\frac{1}{2}}(2)$ and for any bounded (by $M$) random variable it is no more than $M\log^{-\frac{1}{2}}(2)$. For $\mathbf{Exponential(\lambda)}$, it is not a sub-Gaussian random variable, but has sub-exponential norm $\frac{2}{\lambda}$.

For a random vector $a\in \R^n$ we say $a$ is sub-Gaussian if
$$ \|a\|_{\psi_2}:=\sup_{x\in \S^{n-1}} \|\langle a,x\rangle \|_{\psi_2}<\infty, $$
and say $a$ is isotropic if
$$ \E aa^T=I_n.$$
We say a random matrix $A\in \R^{m\times n}$ is isotropic and sub-Gaussian if its rows are independent, isotropic and sub-Gaussian random vectors in $\R^n$. The sub-Gaussian parameter of $A$ is defined as
$$ K:=\max_{1\leq i\leq m}\{\|A_i\|_{\psi_2}: A_i^T \text{ is the $i$-th row of } A\}. $$

For random matrix $A\in \R^{m\times n}$, the isotropic condition guarantees $\frac{1}{\sqrt{m}}A$ will preserve Euclidean norm in expectation.
Some examples of isotropic and sub-Gaussian matrices include matrices with independent and sub-Gaussian entries $A_{ij}$ satisfying $\E A_{ij}=0$ and $\E A_{ij}^2=1$, uniformly subsampled (with replacement and after proper normalization) rows of orthonormal basis or tight frames, etc. \cite{vershynin_2018}.
In the cases of Bernoulli matrices or sparse ternary matrices, which is a generalization of the database-friendly mappings in \cite{achlioptas2003database}, the sub-Gaussian parameter can depend on the signal dimension $n$ if the probability of an entry being nonzero is $n$-dependent.

In the line of research regarding sub-Gaussian random mappings, Liaw et al. \cite{liaw2017simple} showed that for isotropic and sub-Gaussian mapping $A$ with sub-Gaussian parameter $K$, let $T\subset\R^n$, then we have with high probability,
\begin{equation}
\label{eq_liaw}
\sup_{x\in T} \left|\frac{1}{\sqrt{m}}\|Ax\|_2-\|x\|_2\right| \leq \frac{K^2\cdot O(w(T)+\mathrm{rad}(T))}{\sqrt{m}}.
\end{equation}
Here $w(T)$ is the Gaussian width given by
\[w(T):=\E \sup_{y\in T} \langle g, y \rangle \text{ where }g\sim \mathbf{Normal}(0,I_n),\]
and $\mathrm{rad}(T)$ is given by
$$\mathrm{rad}(T):=\sup_{y\in T} \|y\|_2,$$ which is the radius when $T$ is symmetric.

Gaussian width measures the complexity of a set. In particular, denote $\text{cone}(T):=\{tx:t\geq0, x\in T\}$, then $w^2(\text{cone}(T)\cap \S^{n-1})$ is a meaningful approximation for dimension \cite{candes2014mathematics, 2017oymak}.
Generally $\text{rad}(T)$ is also dominated by $w(T)$. For example, if $0\in T$, then by Jensen's inequality, $$w(T) = \E \sup_{y\in T} \lp \max\{ \inp{g}{y},0\} \rp \geq \sup_{y\in T} \E \max\{ \inp{g}{y},0\} = \text{rad}(T)/\sqrt{2\pi}.$$
In such case, \eqref{eq_liaw} implies that with high probability, $\frac{1}{\sqrt{m}}A$ is a near isometry on $T$ whenever $m\geq CK^4 w^2(T)$ for some constant $C$.

The dependency on $w(T)$ in \eqref{eq_liaw} is optimal. This is easy to see when $m=1$ and $A$ has i.i.d.  $\mathbf{Normal}(0,1)$ entries. But when it comes to the dependency on the sub-Gaussian parameter $K$, whether the $K^2$ factor can be improved is a question raised but left unanswered in \cite{liaw2017simple}. 
Other important works regarding this type of bounds are either not explicit \cite{2005Klartag, 2017oymak} or at least of the same order $K^2$ \cite{2007Mendelson, Dirksen2016dimension, 2015dirksen}. 

In this article, we refine this dependency on the sub-Gaussian parameter from $K^2$ to the optimal $K\sqrt{\log K}$. This enhances the concentration bound substantially when the sub-Gaussian mapping is not well-behaved, for example, when $K$ increases together with the signal dimension.
We also relax the row-independent requirement by considering random mappings in the form of $BA$ where $B$ is an arbitrary matrix and $A$ is mean zero, isotropic and sub-Gaussian. The mean zero assumption is additional when comparing to the assumptions for \eqref{eq_liaw}, and not needed when $B$ is only diagonal. However, it is necessary for arbitrary $B$. 
Our bound is broadly applicable since it only require these properties from the random matrix $A$ without any other assumptions.

Now we state our main theorem. In the following, $\|B\|_F$ and $\|B\|$ denote Frobenius and operator norm of $B$ respectively. The matrix $B\in \R^{l\times m}$ is diagonal means that the only possible non-zero entries are $B_{ii}$ where $1\leq i\leq \min\{l,m\}$.
\begin{theorem}
\label{theorem_mainBA}
Let $B\in\R^{l \times m}$ be a fixed matrix, let $A\in\R^{m\times n}$ be a mean zero, isotropic and sub-Gaussian matrix with sub-Gaussian parameter $K$ and let $T\subset \R^n$ be a bounded set. Then
$$ \E\sup_{x\in T} \left|\|BAx\|_2-\|B\|_F\|x\|_2\right|\leq CK\sqrt{\log K}\,\|B\|\left[w(T)+\mathrm{rad}(T)\right], $$
and for any $u\geq 0$, with probability at least $1-3e^{-u^2}$,
$$ \sup_{x\in T} \left|\|BAx\|_2-\|B\|_F\|x\|_2\right|\leq CK\sqrt{\log K}\,\|B\| \left[w(T)+u\cdot\mathrm{rad}(T)\right]. $$
Here $C$ is an absolute constant. Furthermore, when $B$ is a diagonal matrix, random matrix $A$ only need to be isotropic and sub-Gaussian with sub-Gaussian parameter $K$ for the conclusions to hold.
\end{theorem}
When $B$ is the identity matrix, we have the following corollary.
\begin{corollary}
\label{theorem_main}
Let $A\in\R^{m\times n}$ be an isotropic and sub-Gaussian matrix with sub-Gaussian parameter $K$, and let $T\subset \R^n$ be a bounded set. Then
$$ \E\sup_{x\in T} \left|\|Ax\|_2-\sqrt{m}\|x\|_2\right|\leq CK\sqrt{\log K}\,\left[w(T)+\mathrm{rad}(T)\right], $$
and for any $u\geq 0$, with probability at least $1-3e^{-u^2}$,
$$ \sup_{x\in T} \left|\|Ax\|_2-\sqrt{m}\|x\|_2\right|\leq CK\sqrt{\log K}\,\left[w(T)+u\cdot\mathrm{rad}(T)\right] $$
\end{corollary}

In general, the high probability concentration bound for supremum in \Cref{theorem_mainBA} (and \Cref{theorem_main}) is optimal up to constants.
As an example where $K\sqrt{\log K}$ dependency is optimal, consider the scaled Bernoulli distribution as in \Cref{prop_tightness_ex}. In this case, we can let $T=\{(1,0,\dots,0)^T\}$ as a singleton, let $B$ be identity and let $A$ be any isotropic matrix whose entries are independent with $A_{ij}^2$ following the same distribution as $X_i^2$ in \Cref{prop_tightness_ex}. Note that such $A$ is not unique, and we can have different values for $\E A$ by assigning different (arbitrary or random) signs for $A_{ij}$ while keeping $A$ isotropic -- for example, if all $A_{ij}$ are symmetric, then $\E A=0$; if $A$ has a column which is non-negative with probability 1, and $A_{ij}$ in all other columns are symmetric, then $\E A \neq 0$.
In either case, by \Cref{prop_tightness_ex} (assuming $m\geq K^2\log K$), the sub-Gaussian process here has $\psi_2$-norm at least $cK\sqrt{\log K}$ as $K \to \infty$, this implies that the $K\sqrt{\log K}$ factor in our concentration bound cannot be improved (because up to constants, a sub-Gaussian concentration bound is also an upper bound on $\psi_2$-norm).

The $\|B\|$ factor is optimal and this is easy to see when $B$ has non-zero singular values being all equal (because the statement should be invariant under scaling for $B$). We also give another example below in which the singular values are not all equal. As mentioned after \Cref{eq_liaw}, $\mathrm{rad}(T)$ is generally dominated by $w(T)$ and the dependency on $w(T)$ is optimal as well.

Assuming $\mathrm{rad}(T)$ is dominated by $w(T)$, \Cref{theorem_mainBA} then implies that with high probability, matrix $ \frac{1}{\|B\|_F} BA$ is a near isometry on $T$ whenever the stable rank of $B$
$$ 
\mathrm{sr}(B):=\frac{\|B\|_F^2}{\|B\|^2} \geq CK^2\log K \,w^2(T).
	$$
This result recovers \eqref{eq_liaw} with improved dependency on $K$ when $B=I_m$.

\Cref{theorem_mainBA} can fail for some $B$ if $\E A \neq 0$. For example, let $B$ be the all ones matrix, i.e. $B_{ij}=1$ for $1\leq i\leq l$ and $1\leq j\leq m$, then $\|B\|_F=\|B\|=\sqrt{lm}$. Suppose $A$ has independent entries where
\[
\left\{ 
\begin{array}{ll}
A_{ij} \sim \mathbf{Normal}(0,1), & 1\leq i\leq m \,\text{ and }\, 2\leq j\leq n \\
A_{i1}=|g_i|, & g_i\sim \mathbf{Normal}(0,1)
\end{array}
\right.
\]
It is easy to verify that $\E A\neq 0$ and $A$ has isotropic rows $A_i^T$. Moreover, for any $y=(y_1,\dots,y_n)\in \S^{n-1}$, notice that
$$\inp{A_i}{y}\sim \sqrt{1-y_1^2}\cdot \mathbf{Normal}(0,1) + y_1|g_i|.$$
So using triangle inequality for the $\psi_2$-norm and inequality $\sqrt{1-y_1^2}+|y_1|\leq \sqrt{2}$, we get the sub-Gaussian parameter of $A$ is no more than $\sqrt{2}\|g_i\|_{\psi_2}=\sqrt{16/3}$.\\
Let $x=(1,0,\dots,0)^T$ and $T=\{x\}$. Since $Ax=(|g_1|,\dots,|g_m|)^T$, we have
\begin{align*}
\E \left|\|BAx\|_2-\|B\|_F\|x\|_2\right| & \geq \E \|BAx\|_2 - \|B\|_F\|x\|_2 \\
& = \E \lp \sqrt{l} \sum_{i=1}^m |g_i| \rp - \sqrt{lm} \\
&= \sqrt{lm} \lp \sqrt{2m/\pi} -1 \rp .
\end{align*}
On the other hand,
$\|B\|\left[w(T)+\mathrm{rad}(T)\right]=\sqrt{lm}$. So in this case, \Cref{theorem_mainBA} does not hold when $m$ is sufficiently large.

As an example demonstrating $\|B\|$ is optimal in general, consider the case when $T=\{x\}\subset \S^{n-1}$, $A$ is standard Gaussian so that $g:=Ax\sim \mathbf{Normal}(0,I_m)$ and $B=\mathrm{diag}(\tau,1,\dots,1)$ where $\tau>0$. Also let $g_i$ be the coordinates of $g$, then
\begin{align*}
\E \left|\|BAx\|_2-\|B\|_F\|x\|_2\right| & \geq \|B\|_F\|x\|_2 - \E \|BAx\|_2 \\
&= \sqrt{\tau^2+m-1}-\E \sqrt{\tau^2g_1^2+{\textstyle\sum_{i\geq 2}g_i^2}} \\
&\geq \sqrt{\tau^2+m-1} -\E \lp \tau |g_1| + \sqrt{\textstyle\sum_{i\geq 2}g_i^2} \rp \\
&\geq \sqrt{\tau^2+m-1} -\tau\sqrt{2/\pi}-\sqrt{m-1} \\
&=\tau \lp \sqrt{1+\frac{m-1}{\tau^2}}-\sqrt{\frac{2}{\pi}}-\sqrt{\frac{m-1}{\tau^2}} \rp
\end{align*}
where we used Jensen's inequality in the second last line. This estimate is in the order of $\tau=\|B\|$ when $\tau>C\sqrt{m}$ with some constant $C$ large enough.

We make one more technical remark that the $\sqrt{\log K}$ factor here is well-defined. In fact, the isotropic and sub-Gaussian conditions of $A$ guarantee that $K$ is bounded below from $1$. To see this, let $X:=Ax$ for some $x\in \S^{n-1}$, then $X$ has independent coordinates $X_i$ satisfying $\E X_i^2=1$ and $\|X_i\|_{\psi_2}\leq K$. Also let $K_0:=\sqrt{1/ \log 2}\approx 1.201$, from
\begin{equation}
\label{eq_unitVar65}
\E \exp(X_i^2/K_0^2) = \sum_{n\geq 0} \frac{\E X_i^{2n}}{n!K_0^{2n}} \geq \sum_{n\geq 0} \frac{1}{n!K_0^{2n}} = e^{1/K_0^2} = 2
\end{equation}
we can conclude that $K\geq K_0$ and the equality is achieved when $X_i=1$ a.s.

The proof for \Cref{theorem_mainBA} follows an analogous approach in Liaw et al. \cite{liaw2017simple}. One major difference is that we prove and apply two new concentration inequalities with improved parametric dependency in the sub-Gaussian regime. We believe these inequalities are interesting on their own as an application-oriented concentration inequality.

The first one is a new Bernstein type inequality under bounded first absolute moment condition. This inequality provides a concentration bound for sum of sub-exponential random variables.
\begin{theorem}[New Bernstein's Inequality]
\label{theorem_newbernstein}
Let $a=(a_1,\dots,a_m)$ be a fixed non-zero vector and let $Y_1,\dots,Y_m$ be independent, mean zero sub-exponential random variables satisfying $\E |Y_i|\leq 2$ and $\|Y_i\|_{\psi_1} \leq K_i^2$ $\lp \text{assume }K_i \geq \frac{6}{5}\rp$. Then for every $t \geq 0$ we have
$$
\P\left(\left|\sum_{i=1}^m a_iY_i \right|\geq t\right) \leq 
2\exp [-c\min\left(\frac{t^2}{\sum_{i=1}^m a_i^2 K_i^2\log K_i},\frac{t}{\|a\|_{\infty}K^2\log K}\right)], 
	$$
where $K=\max_iK_i$ and $c$ is an absolute constant.
\end{theorem}
\begin{remark}
\label{remark_bernstein}
\Cref{theorem_newbernstein} remains true (with a different absolute constant $c$) when the $2$ in $\E |Y_i|\leq 2$ is replaced with an arbitrary positive constant (see \Cref{remark_EYalpha} for more detail).
\end{remark}

The second one is a new Hanson-Wright inequality under unit variance condition. This inequality provides a concentration bound for quadratic forms of independent random variables and is more general than the aforementioned Bernstein's inequality. In the literature, results of similar flavor have been obtained \cite{rudelson2013hanson,vu2015random,adamczak2015note,klochkov2018uniform} but under different assumptions. We will give a brief comparison between our result and a few notable ones in \Cref{sec_hansonwright}.

\begin{theorem}[New Hanson-Wright Inequality]
\label{theorem_newhansonwright}
Let $A\in \R^{n\times n}$ be a fixed non-zero matrix and let $X=(X_1,\dots, X_n)\in\R^n$ be a random vector with independent, mean zero, sub-Gaussian coordinates satisfying $\E X_i^2=1$ and $\|X_i\|_{\psi_2}\leq K$. Then for every $t\geq 0$ we have
\[
\P\lp |X^TAX-\E X^TAX| \geq t \rp \leq 2 \exp\left[ -c\min\lp \frac{t^2}{\|A\|_F^2K^2\log K},\frac{t}{\|A\|K^2\log K}\rp \right],
	\]
where $c$ is an absolute constant.
\end{theorem}
\begin{remark}
If $A$ is a diagonal matrix, then \Cref{theorem_newhansonwright} recovers \Cref{theorem_newbernstein} (assuming all $K_i$ are equal) with $Y_i=X_i^2-\E X_i^2$. Therefore this can be viewed as a generalization of the new Bernstein's inequality given in \Cref{theorem_newbernstein}.
\end{remark}

\subsection*{Notations}
We use $\|\cdot\|_2$ for Euclidean norm of vectors, $\|\cdot\|_F$ and $\|\cdot\|$ for Frobenius and operator norm of matrices respectively. We use $\circ$ for Hadamard (entrywise) product. We say $f\lesssim g$ if $f\leq Cg$ for some absolute constant $C$ and say $f\gtrsim g$ if $f\geq Cg$ for some absolute constant $C$.
Typically, $c$ and $C$ denote absolute constants (often $c$ for small ones and $C$ for large ones) which may vary from line to line.

\subsection*{Organization}
The rest of this paper is organized as follows: In \Cref{sec_results_bernstein}, we discuss and prove the new Bernstein's inequality (\Cref{theorem_newbernstein}). In \Cref{sec_hansonwright}, we first discuss and compare the new Hanson-Wright inequality (\Cref{theorem_newhansonwright}) to other known variants of Hanson-Wright inequalities and then prove \Cref{theorem_newhansonwright}. In \Cref{sec_results_main}, we prove our main theorem regarding sub-Gaussian matrices on sets (\Cref{theorem_mainBA}) and give an example to show our tail dependency on $K$ is optimal. In \Cref{sec_applications}, we demonstrate how our result can improve theoretical guarantees of some popular applications such as Johnson-Lindenstrauss embedding, null space property for 0-1 matrices, randomized sketches and blind demodulation. In \Cref{sec_conclusion}, we make a brief conclusion for this paper.

\section{New Bernstein's Inequality}
\label{sec_results_bernstein}
In this section we prove the new Bernstein's inequality \Cref{theorem_newbernstein}.
Let us first recall the standard Bernstein's inequality for sub-exponential random variables \cite[Theorem 2.8.2]{vershynin_2018}, which states that for independent, mean zero, sub-exponential random variables $Y_1, Y_2,\dots, Y_m$ and a vector $a=(a_1,\dots,a_m)\in\R^m$, we have
\begin{equation}
\label{eq-std-bernstein}
\P\left(\left|\sum_{i=1}^m a_iY_i \right|\geq u\right) \leq 
2\exp [-c\min\left(\frac{u^2}{K^4\|a\|_2^2},\frac{u}{K^2\|a\|_{\infty}}\right)], 
\end{equation}
where $K^2=\max_i\|Y_i\|_{\psi_1}$.

Compared to \eqref{eq-std-bernstein}, \Cref{theorem_newbernstein} has an extra assumption on the first absolute moment of $Y_i$ -- namely $\E|Y_i|\leq 2$, but it improves the dependence on $K$ in the sub-Gaussian regime from $K^4$ to $K^2\log K$.
It is worth noting that such extra assumption comes naturally when considering isotropic random matrices/vectors. In fact, let $x$ be a fixed point on the unit sphere and let $a_i$ be isotropic random vectors of the same dimension, then $Y_i:=|\inp{a_i}{x}|^2-1$ is mean zero since $a_i$ is isotropic, and $\E |Y_i|\leq \E |\inp{a_i}{x}|^2+1=2$ by triangle inequality.

\subsection*{Proof of \Cref{theorem_newbernstein}}
We will first bound the moments of $Y_i$, then bound their moment generating functions, and finally use Chernoff method to obtain the desired tail bound.

\step{Step 1: Bounding the moments}
The idea here is to write the moment as an integral and then estimate under the two constraints $\E |Y_i|\leq 2$ and $\|Y_i\|_{\psi_1}\leq K^2$.
\begin{lemma}[Moment Bounds]
\label{lemma_pmoment}
Let $Y$ be a sub-exponential random variable satisfying $\E |Y|\leq 2$ and $\|Y\|_{\psi_1} \leq K^2$ with $ K \geq \frac{6}{5}$. Then
\begin{equation*}
 \E |Y|^{p} \leq C^{p} p^p \left(K^2\log K\right)^{p-1}, \; \forall p\geq 1.
\end{equation*}
\end{lemma}
\begin{proof}
Define $f(t):=  \P(|Y|\geq t) \, e^{t/K^2}$. 
Since $\E |Y|\leq 2$, we have 
\begin{equation}
\int_0^\infty f(t)e^{-t/K^2} dt = \int_0^\infty \P(|Y|\geq t) dt \leq 2.
\label{eq-pmoment-1}
\end{equation}
Also, since $\|Y\|_{\psi_1} \leq K^2$, a change of variable $s=e^{t/K^2}$ gives
\begin{align*}
2  \geq \E \exp(|Y|/K^2) &=\int_0^\infty \P\lp e^{|Y|/K^2} \geq s\rp ds 
=\int_{-\infty}^0 K^{-2}e^{t/K^2} dt + \int_0^\infty K^{-2}f(t) dt.
\end{align*}
Notice that $\int_{-\infty}^0 K^{-2}e^{t/K^2} dt=1$, this becomes
\begin{equation}
\int_0^\infty f(t) \,dt \leq K^2.
\label{eq-pmoment-2}
\end{equation}
For the $p$-th moment of $|Y|$, with a change of variable $s=u^p$, we have
$$ \E |Y|^{p} = \int_0^\infty \P(|Y|^p \geq s) d s 
= \int_0^\infty f(u)e^{-u/K^2} pu^{p-1} du.
$$
We will split this integral into two parts. 

Set $T=6pK^2\log K$. Since $pu^{p-1}$ monotonically increases on $[0,T]$, we have
\begin{align*}
\int_0^T f(u)e^{-u/K^2} pu^{p-1} du \; &\leq\; pT^{p-1}\int_0^T f(u)e^{-u/K^2} du \\
&\overset{\eqref{eq-pmoment-1}}{\leq}\; 2p\left(6pK^2\log K\right)^{p-1}.
\end{align*}
On the other hand, since
$$
 \frac{d}{du} \left(u^{p-1}e^{-u/K^2}\right)=K^{-2}e^{-u/K^2}u^{p-2}\left(K^2(p-1)-u\right),
 	$$
and $T\geq \lp 6\log \frac{6}{5}\rp pK^2>pK^2$ (note that $6\log \frac{6}{5}\approx 1.09$), we can conclude that $u^{p-1}e^{-u/K^2}$ monotonically decreases on $[T, \infty)$. Thus
\begin{align*}
\int_T^\infty f(u)e^{-u/K^2} pu^{p-1} du 
&\leq pT^{p-1}e^{-T/K^2} \int_T^\infty f(u)du \\
&\overset{\eqref{eq-pmoment-2}}{\leq} pT^{p-1}K^{-6p}K^2 \\
&\leq p\left(6pK^2\log K\right)^{p-1}.
\end{align*}
Combining these two parts completes the proof with $C\leq 6$.
\end{proof}

\step{Step 2: Bounding the moment generating function}
Let $Y$ be the random variable as in \Cref{lemma_pmoment}, the moment generating function of $Y$ can be estimated through Taylor series
\begin{align*}
\E \exp(\lambda Y) &= 1+\E Y +\sum_{p\geq 2}\frac{\E (\lambda Y)^p}{p!}  \\
&\leq 1+\sum_{p\geq 2} \frac{|\lambda|^p(C_1p)^p\left(K^2\log K\right)^{p-1}}{p!} \\
&\leq 1+\frac{1}{K^2\log K}\sum_{p\geq 2} \left(C_1e|\lambda| K^2\log K\right)^p.
\end{align*}
Here the first inequality is by \Cref{lemma_pmoment} (with $C_1\leq 6$) and the second inequality uses $p!\geq (p/e)^p$.\\
When $\displaystyle|\lambda|K^2\log K\leq 1/(2C_1e)$, the above summation converges, and we have
\begin{align*}
\E \exp(\lambda Y) 
\leq 1+C_1^2e^2\lambda^2K^2\log K
\leq \exp\left(C_1^2e^2\lambda^2K^2\log K\right),
\end{align*}
where the last inequality uses $1+x\leq e^x$.\\
Hence we have showed 
\begin{equation}
\label{eq_mgfY}
\E\exp(\lambda Y) \leq \exp\left(C_0\lambda^2K^2\log K\right) \text{   when   } |\lambda|K^2\log K\leq c_0
\end{equation}
for absolute constants $C_0=(C_1e)^2$ and $c_0=1/(2C_1e)$.

\step{Step 3: Chernoff bound}
For $\displaystyle \lambda \in \left[0,\frac{c_0}{\|a\|_\infty K^2\log K} \right]$, by Markov's inequality and \Cref{eq_mgfY} we have 
\begin{align}
\P\lp \sum_{i=1}^m a_iY_i \geq u\rp 
&\leq e^{-\lambda u}\E \exp\lp \lambda \sum_{i=1}^m a_i Y_i\rp \nonumber \\
& \leq \exp\left(-\lambda u+\lambda^2C_0\sum_{i=1}^m a_i^2 K_i^2\log K_i\right) \label{eq_bern1}
\end{align}
where $c_0$ and $C_0$ are absolute constants. When we minimize the above expression over $\lambda$, we get the optimal value
$$\lambda_{\mathrm{opt}}=\min\left(\frac{u}{2C_0\sum a_i^2 K_i^2\log K_i},\frac{c_0}{\|a\|_\infty K^2\log K}\right).$$
Next we plug in $\lambda_{\mathrm{opt}}$ into \eqref{eq_bern1} to get
\begin{equation}
\label{eq_bern2}
\P\left(\sum_{i=1}^m a_iY_i \geq u\right) \leq \exp[-\min\left(\frac{u^2}{4C_0\sum a_i^2 K_i^2\log K_i},\frac{c_0u}{2\|a\|_\infty K^2\log K}\right)]. 
\end{equation}
Setting $\displaystyle c=\min\left\{\frac{1}{4C_0},\frac{c_0}{2}\right\}$ in \eqref{eq_bern2}, we obtain the one sided bound
$$ 
\P\left(\sum_{i=1}^m a_iY_i\geq u\right) \leq \exp [-c\min\left(\frac{u^2}{\sum a_i^2 K_i^2\log K_i},\frac{u}{\|a\|_\infty K^2\log K}\right)]. 
	$$
The bound for $\P \lp \sum a_iY_i<-u \rp$ is similarly obtained by considering $-Y_i$ instead of $Y_i$. This completes the proof.

\begin{remark}
\label{remark_EYalpha}
If the random variables $Y_i$ have first absolute moment $\E |Y_i|\leq \alpha$, then the right hand side of \Cref{eq-pmoment-1} becomes $\alpha$ and it is easy to see that \Cref{lemma_pmoment} still holds with $C\leq 6+\alpha$. It follows that the $C_1$ in Step 2 will be no more than $6+\alpha$ and \Cref{theorem_newbernstein} now holds with constant $c=\frac{1}{4(C_1e)^2}\geq \frac{1}{4e^2(6+\alpha)^2}$.
\end{remark}

\section{New Hanson-Wright Inequality}
\label{sec_hansonwright}
Hanson-Wright inequality gives a concentration bound for quadratic forms of random variables. The version in \cite{rudelson2013hanson} states that for a random vector $X=(X_1,\dots,X_n)\in\R^n$ with independent, mean zero, sub-Gaussian coordinates, suppose $\max_i\|X_i\|_{\psi_2}\leq K$ and let $A$ be an $n\times n$ real matrix, then
\begin{equation}
\label{eq-hansonwright}
\P\lp |X^TAX-\E X^TAX| \geq t \rp \leq 2 \exp\left[ -c\min\lp \frac{t^2}{\|A\|_F^2 K^4},\frac{t}{\|A\| K^2}\rp \right].
\end{equation}

In the same spirit as the new Bernstein's inequality, we can improve the tail dependency on $K$ in the sub-Gaussian regime from $K^4$ to $K^2\log K$ under a  further assumption $\E X_i^2=1$ for each $X_i$. This is the new Hanson-Wright inequality \Cref{theorem_newhansonwright}.

It is not difficult to drop the requirement $\E X_i^2=1$ in \Cref{theorem_newhansonwright} by a simple scaling, in which case we have the following corollary.
\begin{corollary}
\label{theorem_hansonwright_alt}
Let $X=(X_1,\dots, X_n)\in\R^n$ be a random vector with independent, mean zero, sub-Gaussian coordinates satisfying $0<\|X_i\|_{\psi_2}\leq K$, then for fixed square matrix $A$,
\[
\P\lp |X^TAX-\E X^TAX| \geq t \rp \leq 2 \exp\left[ -c\min\lp \frac{t^2}{\|A\|_F^2\alpha_2^2\gamma^2K^2\log\frac{K}{\alpha_1}},\frac{t}{\|A\|\gamma^2K^2\log \frac{K}{\alpha_1}}\rp \right].
	\]
where $\alpha_1=\min_i\lp \E X_i^2\rp^{\frac{1}{2}}$, $\alpha_2=\max_i\lp \E X_i^2\rp^{\frac{1}{2}}$ and $\gamma=\alpha_2/\alpha_1$.
\end{corollary}

\begin{proof}
Let $\beta_i:=(\E X_i^2)^\frac{1}{2}$ for $1\leq i\leq n$ and define diagonal matrices
$$
D_\beta:=\text{diag}(\beta_1,\beta_2,\dots,\beta_n), \quad
D_{1/\beta}:=\text{diag}(1/\beta_1,1/\beta_2,\dots,1/\beta_n).
	$$
Then $\tilde{X}:=D_{1/\beta}X$ satisfies the assumption of \Cref{theorem_newhansonwright} with $\E\tilde{X}_i^2=1$ and $\|\tilde{X}_i\|_{\psi_2}\leq K/\beta_i\leq K/\alpha_1$. 
Applying \Cref{theorem_newhansonwright} to $\tilde{X}$ and $\tilde{A}:=D_\beta AD_\beta$ completes the proof.
\end{proof}
 
\subsection*{Comparison with Other Hanson-Wright Inequalities}
Let us first compare \Cref{theorem_hansonwright_alt} to the standard Hanson-Wright inequality \eqref{eq-hansonwright} in the case when $\gamma=1$. The concentration bound in \eqref{eq-hansonwright} implies that, with probability at least $1-2e^{-t}$,
\begin{equation}
\label{eq-hw-cmp1}
|X^TAX-\E X^TAX| \lesssim K^2\|A\|_F\sqrt{t}+K^2\|A\|t.
\end{equation}
Meanwhile, \Cref{theorem_hansonwright_alt} implies that
\begin{equation}
\label{eq-hw-cmp2}
|X^TAX-\E X^TAX| \lesssim \alpha K \sqrt{\log (K/\alpha)} \|A\|_F\sqrt{t}+K^2\log(K/\alpha)\|A\|t
\end{equation}
where $\alpha=(\E X_i)^\frac{1}{2}$. Note that $\alpha\leq \|X_i\|_{\psi_2}\leq K$, so this bound improves the parameter dependence (up to a log factor) in the sub-Gaussian regime from $K^2$ to $\alpha K$. Such improvement can be significant when $\alpha$ is far less than $K$.

Other variants of Hanson-Wright inequality have appeared in literature with similar improvements \cite{adamczak2015note, vu2015random}. In particular, one of the results by Adamczak \cite{adamczak2015note} works under the assumption that $X$ satisfies the convex concentration property with constant $\tilde{K}$, that is, for every 1-Lipschitz convex function $\varphi : \R^n\to \R$, we always have $\E |\varphi(X)|< \infty$ and 
$$ \P(|\varphi(X)-\E\varphi(X)|\geq t)\leq 2\exp(-t^2/\tilde{K}^2) \;\text{ for any } t\geq 0.$$
Then under such assumption,
\begin{equation}
\label{eq-hw-cmp3}
|X^TAX-\E X^TAX| \lesssim \tilde{K} \sqrt{\|\text{Cov}(X)\|} \|A\|_F\sqrt{t}+\tilde{K}^2\|A\|t
\end{equation}
where $\text{Cov}(X)$ is the covariance matrix of $X$. When $X$ has independent and mean zero coordinates, $\|\text{Cov}(X)\|=\max_i\E X_i^2$. 
However, the convex concentration property is not the same as sub-Gaussianity. More precisely, while it is true that $\tilde{K}$ is independent of dimension when $X$ has i.i.d. coordinates which are bounded almost surely \cite{samson2000concentration}, this can fail when the boundedness assumption of $X_i$ is replaced by sub-Gaussianity (i.e. $\tilde{K}$ could depend on the dimension of $X$ when $X_i$ are i.i.d. and sub-Gaussian) \cite{adamczak2005logarithmic, hitczenko1998hypercontractivity}. Therefore, the bound in \eqref{eq-hw-cmp3} does not imply \eqref{eq-hw-cmp2} nor \eqref{eq-hw-cmp1} in general.

In a more recent paper by Klochkov and Zhivotovskiy \cite{klochkov2018uniform}, the authors proved a uniform version of the Hanson-Wright inequality which, when applying to a single matrix under the same assumption as \Cref{theorem_hansonwright_alt}, yields the following bound:
\begin{equation}
\label{eq-hw-cmp4}
X^TAX-\E X^TAX \lesssim M\E\|AX\|_2 \sqrt{t}+M^2\|A\|t
\end{equation}
where $M=\|\max_i |X_i|\|_{\psi_2}$. This bound also improves \eqref{eq-hw-cmp1} in some cases as demonstrated in \cite{klochkov2018uniform}. We shall compare this bound to \eqref{eq-hw-cmp2} in the sub-Gaussian regime. On one hand, Jensen's inequality tells us that the $\E \|AX\|_2$ factor in \eqref{eq-hw-cmp4} is bounded by the $\alpha \|A\|_F$ factor in \eqref{eq-hw-cmp2}. On the other, the factor $M$ in \eqref{eq-hw-cmp4} is only bounded by $M \lesssim K\sqrt{\log n}$, which could depend on dimension $n$. Moreover, \eqref{eq-hw-cmp4} only provides a one-sided bound instead of two-sided concentration bounds like \Cref{eq-hw-cmp1,eq-hw-cmp2,eq-hw-cmp3}.

 \subsection*{Proof of \Cref{theorem_newhansonwright}}
 The main idea of proof is similar to \cite{rudelson2013hanson}, that is to divide the sum into diagonal and off-diagonal, then bound the moment generating function of the latter through a decoupling and comparison argument. However, there are two significant differences. The first difference is the random variables used for comparison. We will use scaled Bernoulli multiplied by standard Gaussian in order to preserve the condition of second moment being 1. Using such random variables also leads to challenges in bounding the moment generating function, which is the second difference.
 
Now we proceed with the proof.
For any $t>0$, let $$p:= \P\lp |X^TAX-\E X^TAX| \geq t \rp$$ be the the tail probability we want to bound. Let $A_1:=\text{diag}(A)$ be the diagonal of $A$ and let $A_2:=A-A_1$. Then
$$
p \leq \P\lp |X^TA_1X-\E X^TA_1X| \geq \frac{t}{2} \rp + \P\lp |X^TA_2X-\E X^TA_2X| \geq \frac{t}{2} \rp =:p_1+p_2.
	$$
We will seek bounds for $p_1$ and $p_2$.
\step{Step 1: The diagonal sum}
The bound for $p_1$ is given by our new Bernstein's inequality. Notice that 
$$X^TA_1X-\E X^TA_1X=\sum a_{ii}(X_i^2-1),$$
where $\E |X_i^2-1|\leq 2$ and $\|X_i^2-1\|_{\psi_1}\leq C\|X_i^2\|_{\psi_1}\leq CK^2$.
So by \Cref{theorem_newbernstein} (note that $K\geq 6/5$ as shown in \Cref{eq_unitVar65}) and the simple relationships between the norms of $A_1$ and $A$, we have
\begin{equation*}
p_1\leq 2 \exp\left[ -c\min\lp \frac{t^2}{\|A\|_F^2K^2\log K},\frac{t}{\|A\|K^2\log K}\rp \right].
\end{equation*}

\step{Step 2: Decoupling}
To bound $p_2$, we will derive a bound for the moment generating function of $X^TA_2X$. Let $X'$ be an independent copy of $X$, then 
$$ \E \exp(\lambda X^TA_2X) \leq \E_{X'} \E_{X} \exp(4\lambda X^TAX'). $$
The above follows directly from the following decoupling lemma.
\begin{lemma}[Decoupling \cite{vershynin_2018}] 
\label{lemma_HW_decoupling} 
Let $A=(a_{ij})$ be a fixed $n\times n$ matrix, and let $X=(X_1,\dots, X_n)\in\R^n$ be a random vector with independent mean zero coordinates. Then for every convex function $F:\R\rightarrow \R$,
 $$
 \E F \lp \sum_{i\neq j} a_{ij} X_iX_j \rp \leq \E F \lp 4X^TAX' \rp,
 	$$
where $X'$ is an independent copy of $X$.
 \end{lemma}
 See Theorem 6.1.1 and Remark 6.1.3 in \cite{vershynin_2018} for a proof of \Cref{lemma_HW_decoupling}.

\step{Step 3: Comparison}
We will compare $X$ (and $X'$) to scaled Bernoulli multiplied entrywise by standard Gaussian. But first let us look at the case of a single variable through the following lemma. Note here the Bernoulli parameter $(2L)^{-2}<1$ since $K\geq \sqrt{1/\log 2}$, as shown in \eqref{eq_unitVar65}.
\begin{lemma}
\label{lemma_HW_lm2}
For random variable $Z\in\R$, if $\E Z=0$, $\E Z^2=1$ and $\|Z\|_{\psi_2}\leq K$, then
$$\E \exp(tZ)\leq \E_{r,g}\exp(Ctrg),\quad \forall t\in\R $$ 
where $g\sim\mathbf{Normal}(0,1)$, $r^2\sim (2L)^2\cdot \mathbf{Bernoulli}((2L)^{-2})$ and $L^2=K^2\log K$.
\end{lemma}

\begin{proof}
Using the inequality $e^x\leq x+\cosh(2x)$, which is true for all $x\in\R$ (see \Cref{appendix_b}), we have
$$
\E\exp(tZ) \leq \E tZ+\E\cosh (2tZ) =0+1+\sum_{q\geq 1}\frac{(2t)^{2q}}{(2q)!}\E Z^{2q}.
	$$
By \Cref{lemma_pmoment} we know that $\E Z^{2q}\leq C_0^q q^qL^{2q-2}$ for any positive integer $q$ and some absolute constant $C_0$, hence
$$
\E\exp(tZ) \leq 1+\sum_{q\geq 1} \frac{4^qt^{2q}}{(2q)!}C_0^q q^q L^{2q-2} \leq 1+\sum_{q\geq 1} \frac{(4C_0)^qt^{2q}}{q!}L^{2q-2}.
	$$
On the other hand, a direct calculation gives
$$
\E_{r,g}\exp(Ctrg)=\E_r\exp(\frac{1}{2}C^2t^2r^2)=1+\sum_{q\geq 1}\frac{(C^2/2)^qt^{2q}}{q!}\E r^{2q} \geq 1+\sum_{q\geq 1} \frac{(C^2/2)^qt^{2q}}{q!}L^{2q-2}.
	$$
Choosing any $C$ such that $C^2\geq 8C_0$ completes the proof.
\end{proof}

Now let $g,r\in \R^n$ be random vectors such that $g\sim \mathbf{Normal}(0,I_n)$ and $r$ has entries $r_i^2\overset{i.i.d}{\sim} (2L)^2\cdot \mathbf{Bernoulli}((2L)^{-2})$ where $L^2=K^2\log K$. Also let $g'$ and $r'$ be independent copies of $g$ and $r$. 
Let $\alpha$ be any vector in $\R^n$, by \Cref{lemma_HW_lm2} and independence we have
\begin{equation}
\label{eq-HW-compare}
\E_X\exp(\alpha^TX)
=\prod_j\E_{X_j}\exp(\alpha_jX_j) \leq \prod_j \E_{r_j,g_j}\exp(C\alpha_j\, r_jg_j) 
= \E_{r,g}\exp(C\alpha^T(r\circ g)).
\end{equation}
Note the above also holds for $\E_{X'}\exp(\alpha^TX')$. Therefore
\begin{align*}
\E_{X'} \E_{X} \exp(4\lambda X^TAX') 
&\leq \E_{X'} \E_{r,g} \exp(C\lambda (r\circ g)^TAX') \\
&=  \E_{r,g} \E_{X'} \exp(C\lambda (r\circ g)^TAX') \\
&\leq  \E_{r,g} \E_{r',g'} \exp(C\lambda(r\circ g)^T A(r'\circ g')) \\
&= \E\exp( C\lambda g^TRAR'g' )
\end{align*}
where $R:=\text{diag}(r)$ and $R':=\text{diag}(r')$. Here the two inequalities are repeated applications of \Cref{eq-HW-compare}.

\step{Step 4: Moment generating function of $g^TRAR'g'$}
Denote $\sigma_i=\sigma_i(RAR')$ the singular values of matrix $RAR'$. From the rotation invariance of $g$ and $g'$ we have
$$
\E \exp(\lambda g^TRAR'g' ) 
= \E_{r,r'} \E_{g,g'}\exp(\sum_{i=1}^n \lambda \sigma_ig_ig_i').
	$$
For standard normal random variables $g_i$ and $g_i'$,
\begin{equation*}
\E_{g_i,g_i'} \exp(\eta g_i g_i')=\E_{g_i}\exp(\frac{1}{2}\eta^2g_i^2)=(1-\eta^2)^{-\frac{1}{2}}
\leq \exp(\eta^2) \;\text{ whenever }\; \eta^2<\frac{1}{2},
\end{equation*}
where the inequality uses $(1-x)^{-\frac{1}{2}}\leq e^x$ when $x\in[0,\frac{1}{2})$ (see \Cref{appendix_b}).\\
Also, note that $\sigma_i\leq \|RAR'\|\leq 4L^2\|A\|$, so if $\displaystyle \lambda^2 <\frac{1}{32L^4\|A\|^2}$ we have
$$
\E \exp(\lambda g^TRAR'g' ) 
\leq \E_{r,r'} \exp(\sum_{i=1}^n \lambda^2\sigma_i^2) 
= \E_{r,r'} \exp(\lambda^2\|RAR'\|_F^2).
	$$
Next, use the following \Cref{lemma_HW_lm4} (with $\eta=16\lambda^2L^4$ and $p=\frac{1}{4L^{2}}$) to bound the moment generating function of $\|RAR'\|_F^2$ and we obtain
$$
\E \exp(\lambda g^TRAR'g' ) 
\leq \exp(8\lambda^2L^2\|A\|_F^2)
\;\text{ when }\; \lambda^2 <\frac{1}{32L^4\|A\|^2}.
	$$

\begin{lemma}
\label{lemma_HW_lm4}
Let $D$ be a diagonal random matrix with i.i.d. entries $D_{ii}=d_i\sim \mathbf{Bernoulli}(p)$, and let $D'$ be an independent copy of $D$. Given a fixed matrix $A$, then
$$
\E\exp(\eta\|DAD'\|_F^2) \leq \exp(2p\eta\|A\|_F^2) \;\text{ when }\; 0<\eta\leq\frac{1}{\|A\|^2}.
	$$
\end{lemma}

\begin{proof}
Denote $A_i$ the $i$-th row of $A$. Notice that
$$
\|DAD'\|_F^2\leq \|DA\|_F^2=\sum_i \|A_i\|_2^2 d_i,
	$$
so for $\eta\in \left( 0,\frac{1}{\|A\|^2} \right]$, we have
\begin{align*}
\E\exp(\eta\|DAD'\|_F^2)
&\leq \prod_i \E\exp(\eta \|A_i\|_2^2 d_i) \\
&= \prod_i \lp 1-p+pe^{\eta \|A_i\|_2^2} \rp \\
&\leq \prod_i \lp 1+ 2p\eta\|A_i\|_2^2 \rp  \\
&\leq \exp(2p\eta\|A\|_F^2)
\end{align*}
Here the second last inequality uses $\eta\|A_i\|_2^2\leq \eta\|A\|^2\leq 1$ and $e^x\leq 1+2x$ when $x\in [0,1]$. The last inequality uses $1+x\leq e^x$.
\end{proof}

\step{Step 5: Chernoff bound}
From previous steps we get
$$
\E \exp(\lambda X^TA_2X)
\leq \exp(C\lambda^2L^2\|A\|_F^2)
\;\text{ when }\; \lambda^2 <\frac{c}{L^4\|A\|^2}
	$$
for some absolute constants $C$ and $c$.\\
Notice that $\E X^TA_2X=0$, so by Markov's inequality we have for $0<\lambda \leq \frac{c}{L^2\|A\|}$,
\begin{align*}
\P\lp X^TA_2X-\E X^TA_2X \geq \frac{t}{2} \rp 
& \leq e^{-\lambda t/2}\E\exp(\lambda X^TA_2X) \\
& \leq \exp(-\lambda t/2+C\lambda^2L^2\|A\|_F^2)
\end{align*}
Optimizing this over $\lambda$ (similar to proof of \Cref{theorem_newbernstein}) yields a one sided bound for $p_2$. The other side can then be obtained by considering $-A_2$ (and $-A$) instead of $A_2$. Together they give
\begin{equation*}
p_2\leq 2 \exp\left[ -c\min\lp \frac{t^2}{\|A\|_F^2K^2\log K},\frac{t}{\|A\|K^2\log K}\rp \right].
\end{equation*}

\step{Step 6: The bound for $p$}
Lastly, since $p\leq \min\{1,p_1+p_2\}$, combining the bounds for $p_1, p_2$ and then applying inequality $\min\{1,4e^{-x}\}\leq 2e^{-x/2}$ (see \Cref{appendix_b}) complete the proof of \Cref{theorem_newhansonwright}.

\section{Sub-Gaussian Matrices on Sets}
\label{sec_results_main}
In this section we prove \Cref{theorem_mainBA} and show that the $K\sqrt{\log K}$ tail dependence on $K$ is optimal. \Cref{subsec_main1} studies the simple case when $T$ consists of only a single point. \Cref{subsec_main2} establishes the technical sub-Gaussian increments lemmas and \Cref{subsec_main3} proves \Cref{theorem_mainBA} through these lemmas and Talagrand's Majorizing Measure Theorem. \Cref{sec_tightness} provides an example through scaled Bernoulli random variables that can be used to show the tightness of $K\sqrt{\log K}$ in our concentration bound.

\subsection{Concentration of Random Vectors}
\label{subsec_main1}
Let $X:=Ax\in \R^m$ with $x\in \S^{n-1}$. The isotropic and sub-Gaussian assumption on $A$ now implies $X$ has independent coordinates satisfying $\E X_i^2=1$ and $\|X_i\|_{\psi_2}\leq K$. Moreover, recall that $K\geq \sqrt{1/\log 2}>1$ from \eqref{eq_unitVar65}.

Lemma 5.3 in \cite{liaw2017simple} states that
$$ \| \|X\|_2-\sqrt{m} \|_{\psi_2} \lesssim K^2 .$$
In other words, $\|Ax\|_2$ has a sub-Gaussian concentration around $\sqrt{m}$. It is worth noting that this concentration is independent of the ambient dimension $m$. We will follow a similar proof idea, but use the new inequalities (\Cref{theorem_newbernstein} and \Cref{theorem_newhansonwright}) to generalize and refine this result.

\begin{theorem}
\label{theorem_concentrationBX}
Let $B$ be a fixed $m\times n$ matrix and let $X=(X_1,\dots, X_n)\in\R^n$ be a random vector with independent sub-Gaussian coordinates satisfying $\E X_i^2=1$ and $\|X_i\|_{\psi_2}\leq K$. If either one of the following conditions further holds:
\begin{itemize}
\setlength\itemsep{0em}
\item[$\mathrm{(a)}$] $X$ is mean zero;
\item[$\mathrm{(b)}$] $m=n$ and $B$ is a diagonal matrix.
\end{itemize}
Then
$$ \left\Vert \|BX\|_2-\|B\|_F \right\Vert_{\psi_2}\leq CK\sqrt{\log K}\|B \|. $$
\end{theorem}
\begin{proof}
The conclusion is trivial if $\|B\|=0$, so we will assume $B$ is non-zero.
\vskip .1in
\noindent (a) Let $A:=B^TB$, then
$$ X^TAX=\|BX\|_2^2,\quad \E X^TAX=\|B\|_F^2, $$
and
$$ \|A\|=\|B\|^2, \quad \|A\|_F\leq \|B^T\|\|B\|_F=\|B\|\|B\|_F. $$
Let $Y:=\|BX\|_2^2-\|B\|_F^2$, by \Cref{theorem_newhansonwright} we have
\begin{equation}
\label{eq-concBX-Y}
\P \lp \left| Y \right| \geq t \rp \leq 2\exp[-c\min\left(\frac{t^2}{\|B\|^2\|B\|_F^2 K^2\log K},\frac{t}{\|B\|^2 K^2\log K}\right)]. 
\end{equation}
Note that for $\alpha, \beta, s \geq 0$,
$$
|\alpha-\beta|\geq s \quad\Rightarrow\quad |\alpha^2-\beta^2|\geq \max\{s^2, s\beta\}.
	$$
(This readily comes from the inequalities $|\alpha^2-\beta^2|\geq |\alpha-\beta|^2$ and $|\alpha^2-\beta^2|\geq |\alpha-\beta|\beta$ whenever $\alpha,\beta\geq 0$.)\\
Let $Z:=\|BX\|_2 -\|B\|_F$, then
$$ \P(|Z|\geq s)\leq \P \lp |Y|\geq \max\{s^2, s\|B\|_F\} \rp. $$
To bound this probability, we observe that
\begin{itemize}
\item[]{\makebox[4cm]{if $0\leq s\leq \|B\|_F$, then\hfill}}
$ \displaystyle 
\P(|Z|\geq s) \leq \P( |Y|\geq s\|B\|_F ) \leq 2\exp(\frac{-cs^2}{\|B\|^2 K^2\log K});\;$ and
\item[]{\makebox[4cm]{if $s\geq \|B\|_F$, then\hfill}}
$ \displaystyle
\P(|Z|\geq s) \leq \P(|Y|\geq s^2) \leq 2\exp(\frac{-cs^2}{\|B\|^2 K^2\log K}).$
\end{itemize}
Combining these two bounds and then using property (b) in \Cref{appendix_psi_alpha_properties} complete the proof.

\vskip .1in
\noindent (b) We will first use Bernstein's inequality to obtain \eqref{eq-concBX-Y}. Denote $b_i:=B_{ii}$ the diagonal entries of $B$, then
$$Y:=\|BX\|_2^2-\|B\|_F^2 =\sum_{i=1}^m b_i^2\lp X_i^2-1 \rp.$$
For random variables $X_i^2-1$, notice that
$$\E |X_i^2-1|\leq 2 \quad\text{ and }\quad \|X_i^2-1\|_{\psi_1}\leq C\|X_i^2\|_{\psi_1}\leq CK^2,$$
where the $\psi_1$-norm estimate is from property (f) in \Cref{appendix_psi_alpha_properties}.
Also, note that $K\geq 6/5$ as shown by \Cref{eq_unitVar65}.
Using \Cref{theorem_newbernstein} and the inequality $\sum b_i^4\leq \lp \max_i b_i^2 \rp \cdot \sum b_i^2 = \|B\|^2\|B\|_F^2,$ we have 
$$ 
\P\left(|Y| \geq t\right) \leq 
2\exp [-c\min\left(\frac{t^2}{\|B\|^2\|B\|_F^2 K^2\log K},\frac{t}{\|B\|^2 K^2\log K}\right)]. 
	$$
The rest of the proof is the same as in (a).
\end{proof}

\subsection{Sub-Gaussian Increments Lemma}
\label{subsec_main2}
A key lemma for \Cref{theorem_mainBA} is to show that the random process $Z_x:=\|BAx\|_2-\|B\|_F \|x\|_2 $ has sub-Gaussian increments. That is, $\|Z_x-Z_y\|_{\psi_2}\leq M\|x-y\|_2$ for some $M$ and for all $x,y\in \R^n$. 
Theorem 1.3 in \cite{liaw2017simple} showed sub-Gaussian increments for $B=I_m$ with $M= CK^2$. Here we improve and generalize this result to any $B$ with $M=C K\sqrt{\log K}\,\|B\|$. From \Cref{prop_tightness_ex}, it is easy to see that this  $K\sqrt{\log K}$ factor is in fact tight for certain $A, B, x$ and $y$.

We will prove two versions of the sub-Gaussian increment lemma. The first one (\Cref{lemma_subG_incrementB}) is for arbitrary $B$, but require the random matrix $A$ to be mean zero. The second one (\Cref{lemma_subG_incrementB_diagonal}) is only for diagonal $B$, but does not require zero mean from $A$.

For \Cref{lemma_subG_incrementB} below, the beginning of the proof follows the argument in \cite{liaw2017simple}, except we will use \Cref{theorem_concentrationBX} for better tail bounds. Later on in the proof, we will use a different approach to bound one of the tail probabilities (i.e. $p_3$) through the new Hanson-Wright inequality (\Cref{theorem_newhansonwright}).

\begin{lemma}
\label{lemma_subG_incrementB}
Let $B\in\R^{l \times m}$ be a fixed matrix and let $A\in\R^{m\times n}$ be a mean zero, isotropic and sub-Gaussian matrix with sub-Gaussian parameter $K$. Then the random process
$$ Z_x:=\|BAx\|_2-\|B\|_F \|x\|_2 $$
has sub-Gaussian increments with
$$ \| Z_x-Z_y\|_{\psi_2} \leq CK\sqrt{\log K}\,\|B\| \|x-y\|_2 , \;\; \forall x,y\in\R^n.  $$
\end{lemma}

\begin{proof}
The statement is invariant under scaling for $B$. So without loss of generality, we will assume $B$ has operator norm $\|B\|=1$.
\subsubsection*{Step 1: Show sub-Gaussian increments for $x,y\in \S^{n-1}$ on the unit sphere}
Without loss of generality, assume $x\neq y$ and define
$$ p:=\P \lp \frac{|Z_x-Z_y|}{\|x-y\|_2} \geq s \rp = \P \lp \frac{| \|BAx\|_2-\|BAy\|_2 |}{\|x-y\|_2} \geq s \rp.
	$$
We need to bound this tail probability by a Gaussian whose standard deviation is the order of $K\sqrt{\log K}$. Consider the following two cases:	
\begin{itemize}
\item $s\geq 2\|B\|_F$. Denote $u:=\frac{x-y}{\|x-y\|_2}$ and by triangle inequality we have
$$
p\leq\P \lp \frac{ \|BA(x-y)\|_2}{\|x-y\|_2} \geq s \rp = \P \lp \|BAu\|_2\geq s \rp =: p_1.
	$$
	
\item $0<s<2\|B\|_F$. Write $p$ as
$$
p=\P \lp |Z| \geq s(\|BAx\|_2+\|BAy\|_2) \rp \quad\text{ where }\quad Z:=\frac{\|BAx\|_2^2-\|BAy\|_2^2}{\|x-y\|_2}.
	$$
Then
$$
p\leq \P \lp |Z| \geq s\|BAx\|_2 \rp \leq \P \lp \|BAx\|_2\leq \frac{1}{2}\|B\|_F \rp + \P \lp |Z|>\frac{s}{2}\|B\|_F \rp =: p_2+p_3,
	$$
where $p_2$ and $p_3$ denote the first and second summand respectively.
\end{itemize}
Next we derive bounds for $p_1$, $p_2$ and $p_3$.

\vskip .1in
{\bf \noindent Bound for $p_1$}
\vskip .1in
From $s\geq 2\|B\|_F$ we have
$$
p_1 =  \P \lp \|BAu\|_2-\|B\|_F \geq s -\|B\|_F\rp \leq  \P \lp \|BAu\|_2-\|B\|_F \geq \frac{s}{2} \rp.
	$$
Applying \Cref{theorem_concentrationBX} to the random vector $Au$ we get
$$
p_1 \leq 2\exp(-c\frac{s^2}{4K^2\log K}).
	$$
	
\vskip .1in
{\bf \noindent Bound for $p_2$}
\vskip .1in
Applying \Cref{theorem_concentrationBX} to the random vector $Ax$ and note that $\|B\|_F>\frac{1}{2}s$, we get
$$
p_2\leq 2\exp(-c\frac{(\|B\|_F/2)^2}{K^2\log K})\leq 2\exp(-c\frac{s^2}{16K^2\log K}).
	$$	
	
\vskip .1in
{\bf \noindent Bound for $p_3$}
\vskip .1in
Denote $u:=\frac{x-y}{\|x-y\|_2}$ and $v:=x+y$, then $\inp{u}{v}=0$ since $\|x\|_2=\|y\|_2=1$. We can write $Z$ as
$$
Z=\frac{\inp{BA(x-y)}{BA(x+y)}}{\|x-y\|_2}=\inp{BAu}{BAv}.
	$$
Notice that 
\[
2\inp{BAu}{BAv}=\inp{BA(u+v)}{BA(u+v)}-\inp{BAu}{BAu}-\inp{BAv}{BAv}.
	\]
Let us also denote $X_w:=Aw$ for $w\in \R^n$, then from $\E X_wX_w^T=\|w\|_2^2\, I_n$ we have
$$
\E \|BX_w\|_2^2=\E\, \mathrm{tr}(B^TBX_wX_w^T) = \mathrm{tr}(B^TB\,\E \lp X_wX_w^T \rp ) =\|w\|_2^2\|B\|_F^2.
	$$
Thus we can further write $Z$ as
 \begin{align*}
 Z& = \frac{1}{2}\|BX_{u+v}\|_2^2 - \frac{1}{2}\|BX_u\|_2^2 - \frac{1}{2}\|BX_v\|_2^2 \\
 & = \frac{1}{2} \lp \|BX_{u+v}\|_2^2 -\E \|BX_{u+v}\|_2^2 \rp - \frac{1}{2} \lp \|BX_u\|_2^2-\E \|BX_u\|_2^2 \rp  \notag \\ &\qquad
 - \frac{1}{2} \lp \|BX_v\|_2^2- \E \|BX_v\|_2^2 \rp \\
 & = \frac{1}{2}Y_{u+v} -\frac{1}{2}Y_u -\frac{1}{2}Y_v.
 \end{align*}
where the second equality uses the fact that $Z$ is mean zero and in the last equality $Y_w:=\|BX_w\|_2^2- \E \|BX_w\|_2^2$. Therefore
\begin{align*}
p_3 &= \P (|Y_{u+v}-Y_u-Y_v|>s\|B\|_F) \\
& \leq \P \lp |Y_{u+v}|+|Y_u|+|Y_v|>s\|B\|_F \rp \\
&\leq \P \lp |Y_{u+v}| \geq \frac{s}{2}\|B\|_F \rp + \P \lp |Y_u|+|Y_v| > \frac{s}{2}\|B\|_F \rp \\
&\leq \P \lp |Y_{u+v}| \geq \frac{s}{2}\|B\|_F \rp + \P \lp |Y_u| \geq \lp 1-\frac{1}{8}\|v\|_2^2 \rp \frac{s}{2}\|B\|_F \rp \notag \\ &\qquad
+\P \lp |Y_v| > \frac{1}{8}\|v\|_2^2\cdot \frac{s}{2}\|B\|_F \rp \\
&=: p_4+p_5+p_6.
\end{align*}
We will bound $p_4,\,p_5$ and $p_6$ through the new Hanson-Wright inequality (\Cref{theorem_newhansonwright}).

For any non-zero vector $w$, define $\bar{w}:=\frac{w}{\|w\|_2}$. It is easy to see that $X_w=\|w\|_2 X_{\bar{w}}$ and 
$Y_w=\|w\|_2^2Y_{\bar{w}}$. Also note that
$$ \|B^TB\|=\|B\|^2=1, \quad \|B^TB\|_F\leq \|B^T\|\|B\|_F=\|B\|\|B\|_F=\|B\|_F, $$
so by \Cref{theorem_newhansonwright} we have 
\begin{align*}
\P \lp |Y_{\bar{w}}|\geq r \rp 
 & \leq 2\exp[-c\min\left(\frac{r^2}{\|B\|_F^2 K^2\log K},\frac{r}{K^2\log K}\right)] \\
 &= 2\exp(-c \frac{r^2}{\|B\|_F^2 K^2\log K}) \quad\text{ when }\, 0\leq r \leq \|B\|_F^2.
\end{align*}
Hence for $0\leq t\leq \|w\|_2^2\|B\|_F^2$,
\begin{equation}
\label{eq-boundp456}
\P(|Y_w|\geq t) = \P \lp |Y_{\bar{w}}| \geq \frac{t}{\|w\|_2^2} \rp
\leq 2\exp(\frac{-ct^2}{\|w\|_2^4\|B\|_F^2 K^2\log K}).
\end{equation}
Now we apply \Cref{eq-boundp456} to $p_4,\,p_5$ and $p_6$.

\begin{itemize}
\item For $p_4$. Since $s<2\|B\|_F$ and $\|u+v\|_2=\sqrt{1+\|v\|_2^2}\in [1,\sqrt{5})$, we can conclude that
$$\frac{s}{2}\|B\|_F< \|B\|_F^2\leq \|u+v\|_2^2\|B\|_F^2$$
and therefore
$$ p_4 \leq 2\exp(\frac{-cs^2}{4\|u+v\|_2^4 K^2\log K}) \leq 2\exp(\frac{-cs^2}{100 K^2\log K}). $$

\item For $p_5$. Notice that $\|u\|_2=1$ and $1-\frac{1}{8}\|v\|_2^2 \in (\frac{1}{2}, 1]$, so
$$ p_5\leq \P \lp |Y_u|\geq \frac{s}{4}\|B\|_F \rp \leq 2\exp(\frac{-cs^2}{16 K^2\log K}). 
	$$
	
\item For $p_6$. If $v=0$ (i.e. $x=-y$), then $p_6=\P (0>0)=0$. Now assume $v\neq 0$, then by \eqref{eq-boundp456} we have
$$  p_6=\P \lp |Y_v| > \frac{\|v\|_2^2}{16}s\|B\|_F \rp \leq 2\exp(\frac{-cs^2}{256 K^2\log K}). 
	$$
\end{itemize}

\vskip .1in
{\bf \noindent Putting everything together -- the bound for $p$}
\vskip .1in
So far we have showed that 
\[
p\leq \max\{ p_1, p_2+p_3\}  \quad\text{ and }\quad 
p_3\leq p_4+p_5+p_6,
\]
where $p_i \leq 2\exp(\frac{-cs^2}{K^2\log K})$ for some absolute constant $c$ and $1\leq i\leq 6$. Note $p\leq 1$ and the inequality $\min\{1,8e^{-x}\}\leq 2e^{-x/3}$ (see \Cref{appendix_b}), we get
$$ p\leq \min \left\{ 1, 8\exp(\frac{-cs^2}{K^2\log K})\right\} \leq 2 \exp(\frac{-cs^2}{3K^2\log K}).
	$$

\subsubsection*{Step 2: Show sub-Gaussian increments for all $x$ and $y$}
Without loss of generality, we can assume $\|x\|_2=1$ and $\|y\|_2\geq 1$. Let $\bar{y}:=\frac{y}{\|y\|_2}$ be the projection of $y$ onto unit ball, then by triangle inequality,
$$
\|Z_x-Z_y\|_{\psi_2} \leq \|Z_x-Z_{\bar{y}}\|_{\psi_2} + \|Z_y-Z_{\bar{y}}\|_{\psi_2}=:R_1+R_2.
	$$
Here $R_1$ it is bounded by $CK\sqrt{\log K}\|x-\bar{y}\|_2$ since $x,\bar{y}\in\S^{n-1}$, and $$R_2=\| \lp \|y\|_2-1\rp Z_{\bar{y}}\|_{\psi_2}=\|y-\bar{y}\|_2\|Z_{\bar{y}}\|_{\psi_2}\leq CK\sqrt{\log K}\|y-\bar{y}\|_2,$$
where the first equality uses $Z_y=\|y\|_2Z_{\bar{y}}$, the second equality is true since $\|y\|_2-1=\|y-\bar{y}\|_2$ and the last inequality follows from \Cref{theorem_concentrationBX}. Combining these bounds we get
 $$
 \|Z_x-Z_y\|_{\psi_2} \leq CK\sqrt{\log K} \lp \|x-\bar{y}\|_2+\|y-\bar{y}\|_2 \rp. 
  	$$
Finally, note that $\|x\|_2=1$, so by non-expansiveness of projection, $\|x-\bar{y}\|_2\leq \|x-y\|_2$, and by definition of projection, $\|y-\bar{y}\|_2\leq \|y-x\|_2$. This completes the proof.
\end{proof}

Next we show the second version of sub-Gaussian increment lemma, which requires $B$ to be diagonal and does not need $A$ to be mean zero. The proof is mostly the same as \Cref{lemma_subG_incrementB}, so we will only highlight the differences.
\begin{lemma}
\label{lemma_subG_incrementB_diagonal}
Let $B\in\R^{l \times m}$ be a fixed diagonal matrix and let $A\in\R^{m\times n}$ be a isotropic, sub-Gaussian matrix with sub-Gaussian parameter $K$, then the random process
$$ Z_x:=\|BAx\|_2-\|B\|_F \|x\|_2 $$
has sub-Gaussian increments with
$$ \| Z_x-Z_y\|_{\psi_2} \leq CK\sqrt{\log K}\,  \|B\| \|x-y\|_2, \;\; \forall x,y\in\R^n.  $$
\end{lemma}

\begin{proof}
If $B$ is not a square matrix, we can always add $m-l$ rows of zeros to $B$ (when $l<m$) or remove the last $l-m$ rows of zeros from $B$ (when $l>m$). This will turn $B$ into a $m \times m$ square matrix without changing the values of $\|BAx\|_2$, $\|B\|_F$ and $\|B\|$. So without loss of generality, we can assume $B$ is a square matrix. Also without loss of generality, we can further assume $\|B\|=1$ since the conclusion is invariant under scaling for $B$.

The remaining proof for \Cref{lemma_subG_incrementB_diagonal} is the same as proof for \Cref{lemma_subG_incrementB} except for bounding $p_3$ in Step 1. A bound for $p_3$ here can be obtained through the new Bernstein's inequality (\Cref{theorem_newbernstein}) as detailed below.

Recall that
$$
Z=\inp{BAu}{BAv}=\sum_{i=1}^mb_i^2 \inp{A_i}{u}\inp{A_i}{v}=:\sum_{i=1}^m b_i^2Y_i,
	$$
where $b_i:=B_{ii}$ and $A_i$ is the $i$-th row of $A$. The random variables $Y_i:=\inp{A_i}{u}\inp{A_i}{v}$ are independent, with
$$
\E Y_i=\frac{\E \inp{A_i}{x-y}\inp{A_i}{x+y}}{\|x-y\|_2}=\frac{\E\inp{A_i}{x}^2 -\E\inp{A_i}{y}^2}{\|x-y\|_2}=\frac{1-1}{\|x-y\|_2}=0,
	$$
$$
\E |Y_i|\leq \E \frac{\inp{A_i}{u}^2+\inp{A_i}{v}^2}{2}\leq \frac{1+4}{2}=\frac{5}{2}.
	$$
Here we used $\|x\|_2=\|y\|_2=\|u\|_2=1$, $\|v\|_2\leq 2$, and that $A_i$ is isotropic.
Furthermore, from property (d) in \Cref{appendix_psi_alpha_properties} 
we have
$$
\|Y_i\|_{\psi_1}\leq \|\inp{A_i}{u}\|_{\psi_2}\|\inp{A_i}{v}\|_{\psi_2}\leq K\|u\|_2 \cdot K\|v\|_2\leq 2K^2.
	$$
Therefore by \Cref{theorem_newbernstein} and note that $\sum b_i^4\leq \lp \max_i b_i^2 \rp \cdot \sum b_i^2 = \|B\|_F^2$, we have
$$
\P\lp \left| \sum b_i^2Y_i \right| > t \rp \leq 
2\exp [-c\min\left(\frac{t^2}{\|B\|_F^2 K^2\log K},\frac{t}{K^2\log K}\right)]
	$$
Since $0<s<2\|B\|_F$, we get
$$
p_3=\P\lp \left| Z \right| > \frac{s}{2}\|B\|_F \rp \leq 2\exp(-c\frac{s^2}{4K^2\log K}).
	$$

\end{proof}

\subsection{Proof of \texorpdfstring{\Cref{theorem_mainBA}}{}}
\label{subsec_main3}
\Cref{theorem_mainBA} follows form the sub-Gaussian increments lemmas and Talagrand's Majorizing Measure Theorem. 
Let us first recall the Majorizing Measure Theorem. The following statement is from \cite{liaw2017simple}.
\begin{theorem}[Majorizing Measure Theorem]
\label{theorem_majorm}
Let $(Z_x)_{x\in T}$ be a random process on a bounded set $T\subset \R^n$. Assume that the process has sub-Gaussian increments, that is there exists $M\geq 0$ such that
\[ \|Z_x-Z_y\|_{\psi_2} \leq M\|x-y\|_2 \;\text{ for all }\; x,y \in T. \]
Then
\[ \E \sup_{x,y\in T}|Z_x-Z_y|\leq CM \, \E\sup_{x\in T}\inp{g}{x}, \]
where $g\sim\mathbf{Normal}(0,I_n)$. Moreover, for any $u\geq 1$, the event 
\[ \sup_{x,y\in T}|Z_x-Z_y|\leq CM \lp \E\sup_{x\in T}\inp{g}{x}+u\cdot \mathrm{diam}(T) \rp \]
holds with probability at least $1-e^{-u^2}$. Here $\mathrm{diam}(T):=\sup_{x,y\in T}\|x-y\|_2$.
\end{theorem}
The first part of \Cref{theorem_majorm} can be found in \cite[Theorem 2.4.12]{talagrand2014upper} and the second part can be found in \cite[Theorem 3.2]{2015dirksen}.

\begin{proof}[{\bf Proof of \Cref{theorem_mainBA}}]
Let $Z_x:=\|BAx\|_2-\|B\|_F\|x\|_2$.

For the expectation bound, take an arbitrary $y\in T$, then from triangle inequality we have
\[ \E\sup_{x\in T}|Z_x|\leq \E\sup_{x\in T}|Z_x-Z_y|+\E |Z_y|. \]
Using \Cref{lemma_subG_incrementB} and \Cref{theorem_majorm} (Majorizing Measure Theorem), we get
\[ \E\sup_{x\in T}|Z_x-Z_y| \leq \E\sup_{x,y\in T}|Z_x-Z_y| \lesssim K\sqrt{\log K}\, \|B\| w(T)\]
Using property (e) in \Cref{appendix_psi_alpha_properties} and \Cref{lemma_subG_incrementB}, we get
\[ \E|Z_y| \lesssim \|Z_y\|_{\psi_2} = \|Z_y-Z_0\|_{\psi_2}\lesssim K\sqrt{\log K}\, \|B\|\|y\|_2. \]
Therefore $\E\sup_{x\in T}|Z_x|\leq CK\sqrt{\log K}\, \|B\| \lp w(T)+\mathrm{rad}(T) \rp $.

For the high probability bound, notice that the result is trivial when $u<1$. When $u\geq 1$, fix an arbitrary $y\in T$ and use triangle inequality again to get
\[ \sup_{x\in T}|Z_x|\leq \sup_{x\in T}|Z_x-Z_{y}|+ |Z_{y}| \leq \sup_{x,x'\in T}|Z_x-Z_{x'}|+ |Z_{y}|. \]
Since $\mathrm{diam}(T)\leq 2\,\mathrm{rad}(T)$, applying \Cref{lemma_subG_incrementB} and \Cref{theorem_majorm} we know that the event
\[ \sup_{x,x'\in T}|Z_x-Z_{x'}| \lesssim K\sqrt{\log K} \, \|B\| [ w(T)+u\cdot \mathrm{rad}(T) ] \]
holds with probability at least $1-e^{-u^2}$. \\
To bound $|Z_y|$, again by \Cref{lemma_subG_incrementB}, $\|Z_y\|_{\psi_2}=\|Z_y-Z_0\|_{\psi_2}\leq CK\sqrt{\log K}\,\|B\|\|y\|_2$, so the event
\[ |Z_y| \lesssim uK\sqrt{\log K} \,\|B\| \|y\|_2  \]
holds with probability at least $1-2e^{-u^2}$. Combining these yields the desired high probability bound.

Finally, when $B$ is a diagonal matrix and $A$ is not necessarily mean zero, we can repeat the above argument with \Cref{lemma_subG_incrementB_diagonal} instead of \Cref{lemma_subG_incrementB}. This completes the proof.
\end{proof}

\subsection{An Example for Lower Bound}
\label{sec_tightness}
Here we give an example based on scaled Bernoulli random variables. This example can be used to demonstrate that the $K\sqrt{\log K}$ factor in \Cref{theorem_mainBA} tail bound is optimal in general.
\begin{prop}
\label{prop_tightness_ex}
Let $K\geq 4$ and let $X=(X_1,\dots,X_m)\in\R^m$ be a random vector with independent coordinates such that $\frac{1}{K^2\log K}X_i^2 \sim \mathbf{Bernoulli}\lp \frac{1}{K^2\log K}\rp$,
then $\|X_i\|_{\psi_2}\leq K$. Furthermore, if $m\geq K^2\log K$, then we also have
\begin{equation}
\label{eq-ex-tight}
 \| \|X\|_2-\sqrt{m}\|_{\psi_2} \geq cK\sqrt{\log K}
\end{equation}
for some absolute constant $c> \frac{1}{5}$. Note that this result does not depend on the distribution of the signs of $X_i$.
\end{prop}
Since $X_i$ has the same $\psi_2$-norm as $|X_i|$, and $\|X\|_2$ can be determined by all the $|X_i|$, it is easy to see that the above result does not depend on the distribution of the signs of $X_i$. 
Also, note that the expected number of non-zero coordinates for $X$ is $\frac{m}{K^2\log K}$, so assumption $m\geq K^2\log K$ is requiring this expected number to be at least 1, i.e. a typical realization of $X$ should be non-zero. This assumption can be considered mild.

To prove \Cref{prop_tightness_ex}, we will need a lower bound on Binomial tails. The following estimate is taken from \cite[Lemma 4.7.2]{robert1990ash}:
 \begin{equation}
 \label{eq-binomial-lower-integerk}
 \P\lp \mathbf{Binomial}(m,p)\geq k\rp \geq \frac{1}{\sqrt{8k(1-\frac{k}{m})}}\exp \lp -m D\lp \frac{k}{m} \,\|\, p\rp \rp \;\text{ when }\;  p<\frac{k}{m}<1, \, k\in \mathbb{N}.
 \end{equation}
Here $D(x\|y)$ is the Kullback-Leibler divergence between two Bernoulli distributions with parameters $x$ and $y$ respectively given by
 $$ D(x\|y)=x\log \frac{x}{y} + (1-x)\log\frac{1-x}{1-y}. $$
 Moreover, for $0<y<x<1$,
 \begin{equation}
 \label{eq-KL-mono}
  \frac{\partial}{\partial x}D(x\|y) = \log \frac{x}{y} + \log\frac{1-y}{1-x} > 0.
 \end{equation}
Estimate \eqref{eq-binomial-lower-integerk} is only for integers $k$.
However, using \eqref{eq-binomial-lower-integerk}, we can also prove a variant which is true for all real numbers $k \in (mp+1,m/2)$. This is \Cref{lemma-binomial-lower} below. The proof of this lemma is based on two observations. First, the tail $\P\lp \mathbf{Binomial}(m,p)\geq k\rp$ is a left-continuous decreasing step function, with possible discontinuity at integer points. Second, the right hand side lower bound in \eqref{eq-binomial-lower-integerk} is a decreasing function in $k$ on $[mp,\frac{m}{2}]$.
We will apply \Cref{lemma-binomial-lower} in the proof of \Cref{prop_tightness_ex}.

\begin{lemma}
\label{lemma-binomial-lower}
Assume $p < \frac{1}{4}$ and $mp\geq 1$, then
 \begin{equation*}
 \P\lp \mathbf{Binomial}(m,p)\geq k-1\rp \geq \frac{1}{\sqrt{8k(1-\frac{k}{m})}}\exp \lp -m D\lp \frac{k}{m} \,\|\, p\rp \rp .
 \end{equation*}
for all real numbers $k\in (mp+1, m/2)$.
\end{lemma}
\begin{proof}
Define
$$ q(t):=\P\lp \mathbf{Binomial}(m,p)\geq t\rp  \;\;\text{ and }\;\;
f(t) := \frac{1}{\sqrt{8t(1-\frac{t}{m})}}\exp \lp -m D\lp \frac{t}{m} \,\|\, p\rp \rp.
	$$
Also denote $u=t/m$, we can calculate the derivative of $f(t)$:
\begin{align*}
\frac{df}{dt} &=  \left[ -\frac{1-\frac{2t}{m}}{2\sqrt{8}t^{\frac{3}{2}}(1-\frac{t}{m})^{\frac{3}{2}}} 
+ \frac{1}{\sqrt{8t(1-\frac{t}{m})}} \cdot (-m)\frac{\partial D}{\partial u}\frac{\partial u}{\partial t} \right] 
\exp \lp -m D\lp \frac{t}{m} \,\|\, p\rp \rp \\
&= -\frac{\exp \lp -m D\lp \frac{t}{m} \,\|\, p\rp \rp}{\sqrt{8t(1-\frac{t}{m})}}
\left[ \frac{1-\frac{2t}{m}}{2t(1-\frac{t}{m})} + \frac{\partial D}{\partial u} \right].
\end{align*}
When $t \in [mp, m/2]$, notice $\frac{\partial D}{\partial u}\geq 0$ by \eqref{eq-KL-mono} and we get $\frac{df}{dk}\leq 0$. 
So $f(t)$ monotonically decreases on $ [mp, m/2]$. 

On the other hand, notice that $q(t)$ is a left-continuous decreasing step function, so by the above and \eqref{eq-binomial-lower-integerk} we have
\[
q(k-1)=q(\lceil k-1 \rceil) \geq q(\lfloor k \rfloor) \geq f(\lfloor k \rfloor) \geq f(k), \quad k\in (mp+1,m/2)
\]
where $\lfloor\cdot\rfloor$ and $\lceil\cdot\rceil$ are the floor and ceiling functions.
This completes the proof.
\end{proof}

 \begin{proof}[{\bf Proof of \Cref{prop_tightness_ex}}]
$\|X_i\|_{\psi_2}\leq K$ follows directly from definition since
 $$
 \E \exp(X_i^2/K^2) = \frac{1}{K^2\log K}e^{\log K}+\lp 1-\frac{1}{K^2\log K}\rp e^0 < 2.
 	$$
Let $\lambda>0$, $Z:=\|X\|_2-\sqrt{m}$ and $L^2:=K^2\log K$, with a change of variable $s=\exp(\lambda t/L^2)$ we have
 \begin{align*}
 \E \exp( \lambda Z^2/L^2) 
 &= \int_{0}^\infty \P\left(e^{\lambda Z^2/L^2}\geq s\right) ds \\
 &=\int_0^1 1\, ds + \int_{1}^\infty \P\left(e^{\lambda Z^2/L^2}\geq s\right) ds \\
 &= 1 + \frac{\lambda}{L^2}\int_0^\infty \P(Z^2\geq t)\, e^{\lambda t/L^2} dt.
 \end{align*}
To show \eqref{eq-ex-tight}, we need to find a $\lambda$ such that $\E \exp( \lambda Z^2/L^2)>2 $.
By a change of variable $t=v^2L^2$, it suffices to show
 $$I:=2\lambda \int_0^\infty \P(|Z|\geq vL)\, ve^{\lambda v^2}dv>1 \;\text{ for some } \lambda>0. $$ 
 Let
 $$\alpha:=\frac{\sqrt{m}}{L}\geq 1, \quad 
 \beta_v:=\alpha +v=\frac{\sqrt{m}}{L}+v, \quad 
 \gamma_v:=\frac{\beta_v^2+1}{m}=\lp \frac{1}{L}+\frac{v}{\sqrt{m}}\rp ^2 +\frac{1}{m},
 $$
 then
 \begin{align*}
 \P(|Z|\geq vL) &\geq \P\lp \|X\|_2\geq \sqrt{m}+vL\rp 
 = \P \lp \frac{1}{L^2}\|X\|_2^2\geq \beta_v^2 \rp 
 \end{align*}
 where $\frac{1}{L^2}\|X\|_2^2\sim \mathbf{Binomial}\lp m,\frac{1}{L^2} \rp$. Note that for $v\in [\alpha, 2\alpha]$, we have 
$$\beta_v^2+1>mL^{-2}+1 \quad\text{ and }\quad \beta_v^2+1 \leq 9mL^{-2}+1\leq 10mL^{-2}<m/2.
	$$
So by \Cref{lemma-binomial-lower} (with $k=\beta_v^2+1$) we get
\[
\P \lp |Z|\geq vL \rp
\geq \frac{1}{\sqrt{8(\beta_v^2+1)}} \exp \lp -m D\lp \gamma_v \,\|\, L^{-2} \rp \rp.
\]
Also note that for $v\in [\alpha, 2\alpha]$,
\[
\frac{v^2}{\beta_v^2+1}= \frac{v^2}{(\alpha+v)^2+1} \geq \frac{v^2}{(2v)^2+v^2}=\frac{1}{5}.
\]
Therefore
 \begin{align*}
 I &\geq 
 2\lambda\int_{\alpha}^{2\alpha} \frac{1}{\sqrt{8}} 
 \frac{v}{\sqrt{\beta_v^2+1}}\exp\lp -mD\lp \gamma_v \,\middle\|\, L^{-2} \rp \rp \cdot \exp(\lambda v^2) dv \\
 &\geq \frac{2\lambda}{\sqrt{40}} \int_{\alpha}^{2\alpha} \exp \lp -mD\lp \gamma_v \,\middle\|\, L^{-2} \rp +\lambda v^2 \rp dv \\
 &\geq \frac{\lambda}{\sqrt{10}} \int_{\alpha}^{2\alpha} \exp(-\lambda_0\alpha^2 +\lambda v^2) dv
 \end{align*}
 with $\lambda_0:=10\log 10$. The last inequality above holds because
$\gamma_v \leq \gamma_{2\alpha}=\frac{9}{L^2}+\frac{1}{m}\leq \frac{10}{L^2}$, so it follows from \eqref{eq-KL-mono} that
\[
D\lp \gamma_v \,\middle\|\, L^{-2} \rp \leq
D\lp 10L^{-2} \,\middle\|\, L^{-2}\rp 
 \leq 10L^{-2}\log \frac{10L^{-2}}{L^{-2}} = \frac{\lambda_0}{L^2}.
\]
Choose $\lambda= \lambda_0$ and we get
 $$ I\geq \frac{\lambda_0}{\sqrt{10}}\int_\alpha^{2\alpha} 1\,dv \geq \frac{\lambda_0}{\sqrt{10}} >1.$$ 
 This proves \eqref{eq-ex-tight} with $c=1/\sqrt{\lambda_0}\approx 0.208$.
 \end{proof}

\section{Applications}
\label{sec_applications}
\subsection{Johnson-Lindenstrauss Lemma}
One immediate application of our result is a guarantee for all isotropic and sub-Gaussian matrices as Johnson-Lindenstrauss (JL) embeddings for dimension reduction. We state this JL lemma below. It follows directly form \Cref{theorem_concentrationBX}.

\begin{lemma}
\label{lemma_JL_xy}
Let $A\in\R^{m\times n}$ be an isotropic and sub-Gaussian matrix with sub-Gaussian parameter $K$. If
\begin{equation}
\label{eq-JLlm}
 m \geq CK^2\log K \varepsilon^{-2}\log(1/\delta),
 \end{equation}
then for any $x, y\in\R^n$, with probability at least $1-\delta$ we have
$$ (1-\varepsilon)\|x-y\|_2 \leq \frac{1}{\sqrt{m}} \|A(x-y)\|_2 \leq (1+\varepsilon)\|x-y\|_2. $$
\end{lemma}
\begin{proof}
By scaling we can assume $\|x-y\|_2=1$. By \Cref{theorem_concentrationBX} (with $B=I_m$) we have
 $$ \| \|m^{-\frac{1}{2}}A(x-y)\|_2-1\|_{\psi_2}\leq C m^{-\frac{1}{2}}K\sqrt{\log K},
 $$
 the result then follows from property (a) in \Cref{appendix_psi_alpha_properties}.
 \end{proof}

It is known that the dependence on $\varepsilon$ and $\delta$ in \eqref{eq-JLlm} is optimal for linear mappings \cite{larsen2014johnson}. Using the same example as \Cref{prop_tightness_ex}, we can conclude that (see \Cref{appendix_jlopt}) the dependence on sub-Gaussian parameter $K$ here is optimal as well.
Similar results to \Cref{lemma_JL_xy} have appeared in \cite{matouvsek2008variants,Dirksen2016dimension}, but to the best of our knowledge, the previous known dependence on $K$ was $K^4$.

\subsection{Null Space Property for 0-1 Matrices: Improved Parametric Dependence}

{\it Null space property} (NSP) is a well-known sufficient condition for sparse recovery and can be used to provide guarantees for $\ell_1$-based algorithms such as {\it basis pursuit denoising}. For the robust version, we say a matrix $A\in \R^{m\times n}$ satisfies the {\it $\ell_2$-robust null space property} of order $s$ with parameters $\rho\in (0,1)$ and $\tau >0$, denoted as $\ell_2\text{-rNSP}(s,\rho,\tau)$, if
$$
\|v_S\|_2 \leq \frac{\rho}{\sqrt{s}}\|v_{\overline{S}}\|_1 + \tau \|Av\|_2
	$$
holds for all $v\in \R^n$ and all $S\subset [n]= \{1,2,\dots,n\}$ with size $|S|\leq s$. Here $v_S$ is the restriction of $v$ to set $S$ (that is, $v_S(i)=v(i)$ if $i\in S$ and $v_S(i)=0$ otherwise) and $\overline{S}$ is the complement of $S$. 

Matrices with all entries equal to 0 or 1, called 0-1 matrices, appear in many compressed sensing applications, including group testing \cite{cohen2020multi,shental2020efficient}, compressed imaging \cite{shental2020efficient}, wireless network activity detection \cite{kueng2017robust}.  In some cases, most of the entries of the matrix should be equal to 0.  In particular, in group testing each 1 that appears in a row corresponds with a sample that is being mixed into a pool, but including too many samples causes dilution.  Measurement matrices are designed to be sparse to avoid such dilution \cite{cohen2020multi}.  The traditional way to draw such a 0-1 matrix for group testing is to take each entry as an independent $\mathbf{Bernoulli}(p)$ random variable \cite{cohen2020multi}, although recent research focused on COVID-19 detection has considered other methods to generate the 0-1 matrix \cite{cohen2020multi,shental2020efficient}.  The authors of \cite{kueng2017robust} proved that 0-1 matrices with $\mathbf{Bernoulli}(p)$ entries satisfy $\ell_2$-rNSP under certain conditions, with a proof based on Mendelson’s small ball method.  Using the tools of our paper as a black box allows for a simpler proof of this result, as well as a significantly improvement on the dependence of parameter $p$.

A common way of proving NSP is through the {\it restricted isometry property} (RIP). We say matrix $A\in \R^{m\times n}$ satisfies RIP of order $s$ with parameters $\delta\in (0,1)$, denoted as RIP$(s,\delta)$, if
\[
(1-\delta)\|x\|_2 \leq \|Ax\|_2 \leq (1+\delta)\|x\|_2, \quad \forall x\in \Sigma_s^n
\]
where $\Sigma_s^n$ is the set of all $s$-sparse vectors in $\R^n$ (note that \Cref{theorem_mainBA} can provide such a result by choosing $T=\Sigma_s^n\cap \S^{n-1}$). 
However, for 0-1 matrices, it is necessary to have $m\gtrsim \min\{s^2, n\}$ in order for RIP of order $s$ to hold \cite{chandar2008negative}. This would be highly sub-optimal as the optimal sample complexity for sparse recovery is known to be $Cs\log \frac{en}{s}$, which can also be achieved by many isotropic random matrices \cite{foucart2013mathematical}.

To our knowledge, the best result so far regarding $\ell_2$-rNSP for 0-1 matrices with i.i.d. $\mathbf{Bernoulli}(p)$ entries appeared in \cite{kueng2017robust}. 
While its bound on sample complexity scales optimally in $s$, the dependence on Bernoulli parameter $p$ was at least $\frac{1}{p^2(1-p)^2}$.
In the following we give an alternative proof for such $\ell_2$-rNSP (\Cref{theorem_nsp01}). Our proof yields a sample complexity bound that scales optimally in both $s$ and $p$. We remark that the optimal $\frac{1}{p(1-p)}$ dependence on $p$ is obtained due to the improvement on sub-Gaussian parameters in \Cref{theorem_mainBA}. We also give a brief justification of this optimality after the proof. 

The proof idea for \Cref{theorem_nsp01} is to show RIP (and therefore $\ell_2$-rNSP) for a projected version of the original 0-1 matrix, and then pass the $\ell_2$-rNSP back to the original matrix.

\begin{theorem}
\label{theorem_nsp01}
Fix $\rho \in (0,1)$ and let $A\in \R^{m\times n}$ be a random matrix with i.i.d. $\mathbf{Bernoulli}(p)$ entries where $p\in (0,1)$. If
\begin{equation}
\label{eq-nsp-m}
m\geq C\rho^{-2}\frac{1}{p(1-p)} \lp s\log\frac{en}{s} + u^2 \rp 
\end{equation}
for some absolute constant $C$, then with probability at least $1-3e^{-u^2}$, $A$ satisfies $\ell_2\text{-rNSP}(s,\rho,\tau)$ with $\tau = \frac{2}{\sqrt{m-1}\sqrt{p(1-p)}} $.
\end{theorem}
\begin{proof}
Denote $\mathbf{1}_n$ the all ones column vector in $\R^n$, then $\E A = p\mathbf{1}_m\mathbf{1}_n^T$. Let 
$$\textstyle
\tilde{A}:=\frac{1}{\sqrt{p(1-p)}} \lp A -  \E A \rp,
	$$
it is easy to verify that $\E \tilde{A}_{ij}=0$ and $\E \tilde{A}_{ij}^2=1$. Moreover, by \Cref{prop_nsp_psi2} below, $\tilde{K}^2 \log \tilde{K}\leq \frac{1}{p(1-p)}$ where $\tilde{K}$ is the sub-Gaussian parameter of $\tilde{A}$.

Let $P\in \R^{m\times m}$ be the orthogonal projection matrix onto span$\{\mathbf{1}_m\}^\perp$ and let $T=\Sigma_{2s}^n\cap \S^{n-1}$, then $\|P\|_F=\sqrt{m-1}$ and by \Cref{theorem_mainBA}, we have with probability at least $1-3e^{-u^2}$,
$$
\sup_{x\in T} \left| \frac{1}{\sqrt{m-1}}\|P\tilde{A}x\|_2 - 1 \right| \leq \frac{C}{\sqrt{m-1}}\frac{1}{\sqrt{p(1-p)}} \lp w(T) + u \rp.
	$$
Since $w^2(T)\leq 4s \log \frac{en}{s}$ \cite[Proposition 9.24]{foucart2013mathematical}, by choosing $m$ as in \eqref{eq-nsp-m} with an appropriate absolute constant $C$, we get with probability at least $1-3e^{-u^2}$,
\begin{equation}
\label{eq-nsp-rip}
\sup_{x\in T} \left| \frac{1}{\sqrt{m-1}}\|P\tilde{A}x\|_2 - 1 \right| \leq \frac{1}{2}\rho =: \delta,
\end{equation}
or equivalently, $\frac{1}{\sqrt{m-1}}P\tilde{A}$ satisfies RIP$(2s,\delta)$ where $\delta=\frac{1}{2}\rho \in(0, \frac{1}{2})$.

Note that RIP implies NSP, in particular, using \cite[Theorem 6.13]{foucart2013mathematical} we know \eqref{eq-nsp-rip} implies that $\frac{1}{\sqrt{m-1}}P\tilde{A}$ satisfies $\ell_2$-rNSP$(s,\rho',\tau')$ with
$$
\rho':=\frac{\delta}{\sqrt{1-\delta^2}-\delta/4} < 2\delta =\rho
\quad \text{ and }\quad
\tau':=\frac{\sqrt{1+\delta}}{\sqrt{1-\delta^2}-\delta/4} < 2.
	$$
Notice that $P\tilde{A}= \frac{1}{\sqrt{p(1-p)}} PA$, this robust null space property then reduces to
$$
\|v_S\|_2 \leq \frac{\rho'}{\sqrt{s}}\|v_{\overline{S}}\|_1 + \frac{\tau'}{\sqrt{m-1}\sqrt{p(1-p)}} \|PAv\|_2,
\quad \forall v\in \R^n, S\subset [n] \text{ with } |S|\leq s.
	$$
$\ell_2\text{-rNSP}(s,\rho,\tau)$ for $A$ follows immediately as $\|PAv\|_2\leq \|Av\|_2$.
\end{proof}

\begin{prop}[$\psi_2$-norm of $\tilde{A}_{ij}$ in proof of \Cref{theorem_nsp01}]
\label{prop_nsp_psi2}
Let $p\in (0,1)$ and
\[
X:= \frac{1}{\sqrt{p(1-p)}} \lp \mathbf{Bernoulli}(p) - p \rp
=\left\lbrace 
\begin{array}{lc}
\sqrt{\frac{1-p}{p}} & \text{with probability } p \\
-\sqrt{\frac{p}{1-p}} & \text{with probability } 1-p
\end{array}
\right.
\]
Also let $K$ be the real number such that $K^2\log K=\frac{1}{p(1-p)}$, then $\|X\|_{\psi_2}\leq K$.
\end{prop}
\begin{proof}
Let $q=1-p$ and notice that $q/(pK^2)=q^2\log K$ and $p/(qK^2)=p^2\log K$, we have
\begin{align*}
\E \exp(X^2/K^2) = p e^{q/(pK^2)} + q e^{p/(qK^2)} = p K^{q^2} + qK^{p^2}.
\end{align*}
Since $\frac{1}{pq}\geq 4$, we must have $\log K > \frac{1}{2}$ (otherwise $K^2\log K\leq \frac{e}{2}$), thus $K^2 < \frac{2}{pq}$. So 
\begin{align*}
\E \exp(X^2/K^2) <  p \lp \frac{2}{pq} \rp ^\frac{q^2}{2} + q \lp \frac{2}{pq} \rp ^\frac{p^2}{2} \leq 2
\end{align*}
where the final estimate follows from inequality $(1-x)\lp \frac{2}{x(1-x)} \rp ^ {x^2/2} \leq 1$ when $0<x<1$ (see \Cref{appendix_b} for a proof of this inequality).
\end{proof}

Lastly, we remark that the $\frac{1}{p(1-p)}$ dependence in sample complexity is (up to constants) optimal as $p\to 0$ or $p\to 1$.
In fact, we can show that if $mp<\frac{1}{2}$ or $m(1-p)<\frac{1}{2}$, then matrix $A$ from \Cref{theorem_nsp01} cannot satisfy $\ell_2$-NSP of order 2 (regardless of $\rho$ and $\tau$) with probability at least $\frac{1}{4}$.

To see this, let $q=1-p$, $v=(1,-1,0,\dots,0)^T\in \R^n$, $S=\{1,2\}$ and consider cases:
\begin{itemize}
\item $mp<\frac{1}{2}$. Denote $A_i$ the $i$-th column of $A$, then by Bernoulli's inequality, the probability of $A_i$ being zero is
$$ \P(A_i=\mathbf{0})=(1-p)^m\geq 1-mp > \tfrac{1}{2}. $$
Thus $\P(A_1=A_2=\mathbf{0})> \frac{1}{4}$ and on the event $\{A_1=A_2=\mathbf{0}\}$, we have
$$ \|v_S\|_2=\sqrt{2} \quad\text{ while }\quad  \frac{\rho}{\sqrt{s}}\|v_{\overline{S}}\|_1 + \tau \|Av\|_2=0.
$$
This means with probability at least $\frac{1}{4}$, $A$ cannot satisfy $\ell_2$-rNSP of order 2.

\item $mq<\frac{1}{2}$. Let $B=\mathbf{1}_m\mathbf{1}_n^T -A$ and denote $B_i$ the $i$-th column of $B$.
Notice that $B$ has i.i.d. $\mathbf{Bernoulli}(q)$ entries and $Bv=\mathbf{1}_m\mathbf{1}_n^Tv -Av=-Av$, so by applying the same argument as above to $B$, we can conclude that $A$ does not satisfy $\ell_2$-rNSP of order 2 with probability at least $\frac{1}{4}$.
\end{itemize}

\subsection{Randomized Sketches}
Randomized sketching provide a method for approximating convex programs \cite{2015pilanci,yang2017randomized}. In essence, a randomized sketch reduces the dimension of the original optimization problem through random projections, which can be beneficial in both computational time and memory storage. Following the problem formulation and ideas in \cite{2015pilanci}, consider convex program in the form of
\begin{equation}
\label{origprob}
\min_{x\in\mathcal{C}} f(x):=\|Bx-y\|_2^2,
\end{equation}
where $B\in\R^{n\times d}$, $y\in\R^d$ and $\mathcal{C}\subset \R^d$ is some convex set.
Let $A\in \R^{m\times n}$ be an isotropic and sub-Gaussian matrix and solve instead the convex program
\begin{equation}
\label{sketchprob}
\min_{x\in\mathcal{C}} g(x):=\|A(Bx-y)\|_2^2.
\end{equation}
This is called the "sketched problem". It reduces the dimension from $n$ to $m$ and can be viewed as an approximation to the original problem \eqref{origprob}.
Moreover, say a solution $\hat{x}$ to the sketched problem \eqref{sketchprob} is $\delta$-optimal to the original optimal solution $x^*$ of \eqref{origprob} if
$$ f(\hat{x})\leq (1+\delta)^2f(x^*). $$

Pilanci and Wainwright \cite{2015pilanci} gave a high probability guarantee for $\hat{x}$ being $\delta$-optimal when $m$ is sufficient large. The following \Cref{theorem_delta_opt_guarantee} improves the dependence on $K$ in their guarantee from $K^4$ to $K^2\log K$.
The proof of \Cref{theorem_delta_opt_guarantee} is also more concise thanks to the tools we have developed.
 \begin{theorem}[$\delta$-optimal guarantee]
 \label{theorem_delta_opt_guarantee}
 Let $A$ be an isotropic and sub-Gaussian matrix with sub-Gaussian parameter $K$. For any $\delta\in (0,1)$, if
 $$ m \geq c_0K^2\log K \frac{w^2(B\mathcal{T}\cap \S^{n-1})}{\delta^2}, $$
 then a solution $\hat{x}$ to the sketched problem as given in \eqref{sketchprob} is $\delta$-optimal with probability at least $1-c_1e^{-c_2m\delta^2/(K^2\log K)}$. Here $c_0,c_1,c_2$ are absolute constants and $\mathcal{T}$ is the tangent cone of $\mathcal{C}$ at optimum $x^*$, given by
 $$
 \mathcal{T}:=\text{clconv }\{t(x-x^*): t\geq 0 \text{ and } x\in \mathcal{C}\}
 	$$
where clconv denotes the closed convex hull.
 \end{theorem}
 
We will use an argument similar to \cite{2015pilanci} to prove \Cref{theorem_delta_opt_guarantee}.
First let us state a deterministic result that says $\delta$-optimality can be obtained by controlling two quantities.
 \begin{lemma}[Lemma 1 \cite{2015pilanci}]
 \label{lemma_sketch_ratio}
 For any sketching matrix $A\in\R^{m\times n}$,  let 
 \begin{align*}
 Z_1&:=\inf_{v\in B\mathcal{T}\cap \S^{n-1}} \frac{1}{m}\|Av\|_2^2, \\
 Z_2&:=\sup_{v\in B\mathcal{T}\cap \S^{n-1}} \left| \inp{u}{\lp \frac{1}{m}A^TA-I \rp v} \right|,
 \end{align*}
 where $\mathcal{T}$ is the tangent cone of $\mathcal{C}$ at $x^*$ and $u\in\S^{n-1}$ is an arbitrarily fixed vector. Then
\begin{equation*}
f(\hat{x})\leq \lp 1+2\frac{Z_2}{Z_1}\rp^2 f(x^*). 
\end{equation*}
 \end{lemma}
 Next we show a technical Lemma that will be helpful when estimating $Z_1$ and $Z_2$.
\begin{lemma}
\label{lemma_sketch_lm1}
Let $A$ be an isotropic and sub-Gaussian matrix with sub-Gaussian parameter $K$, and let $T\subset\R^n$ be a set with radius $\text{rad}(T)\leq 2$, then there exists absolute constants $C$ and $c$ such that for any $\delta\in (0,1)$,
\begin{equation*}
\sup_{x\in T}\left| \frac{1}{m}\|Ax\|_2^2-\|x\|_2^2\right|  \leq \delta
\end{equation*}
holds with probability at least $1-3e^{-cm\delta^2/(K^2\log K)}$ provided $m\geq CK^2\log Kw^2(T)/\delta^2$.
\begin{proof}
Denote $L:=K\sqrt{\log K}$. By \Cref{theorem_main} we have
$$
\sup_{x\in T}\left| \frac{1}{\sqrt{m}}\|Ax\|_2-\|x\|_2\right| \leq \frac{C_0L}{\sqrt{m}}\lp w(T)+2u \rp
	$$
with probability at least $1-3e^{-u^2}$ for some absolute constant $C_0$.\\
Take $\displaystyle \delta_0=\frac{1}{5}\delta$ and $\displaystyle m \geq 9C_0^2\frac{L^2w^2(T)}{\delta_0^2}$, also choose $\displaystyle u= \frac{1}{3C_0}\frac{\sqrt{m}\delta_0}{L}$. It follows that the event
$$
\sup_{x\in T}\left| \frac{1}{\sqrt{m}}\|Ax\|_2-\|x\|_2\right| \leq \frac{\delta_0}{3}+\frac{2\delta_0}{3}=\delta_0
	$$
holds with probability at least $1-3e^{-m\delta_0^2/(9C_0^2L^2)}$. On this event,
$$
\sup_{x\in T}\left| \frac{1}{m}\|Ax\|_2^2-\|x\|_2^2 \right| 
\leq (4+\delta_0)\delta_0 \leq 5\delta_0=\delta,
	$$
where we use the estimate 
$ \left| \frac{1}{\sqrt{m}}\|Ax\|_2+\|x\|_2\right| \leq 2\|x\|_2+\delta_0$ for $x\in T$.
\end{proof}
\end{lemma}

\begin{proof}[\bf Proof of \Cref{theorem_delta_opt_guarantee}]
We wish to control the ratio $Z_2/Z_1$ in sight of \Cref{lemma_sketch_ratio}.

By \Cref{lemma_sketch_lm1}, if $m\geq CK^2\log Kw^2(T)/\delta^{2}$, then
$$ \P\lp Z_1\geq 1-\frac{\delta}{2} \rp  \geq 1-3e^{-cm\delta^2/(K^2\log K)}.  $$

Let $T:=B\mathcal{T}\cap \S^{n-1}$ and $Q:=\frac{1}{m}A^TA-I$. Since
$$ 2\inp{u}{Qv}=\inp{u+v}{Q(u+v)}-\inp{u}{Qu}-\inp{v}{Qv}, $$
triangle inequality gives
$$
Z_2\leq \frac{1}{2}\sup_{x\in u+T}|\inp{x}{Qx}| + \frac{1}{2}\sup_{x\in \{u\}}|\inp{x}{Qx}| + \frac{1}{2}\sup_{x\in T}|\inp{x}{Qx}|
=:Z_2^{(1)}+Z_2^{(2)}+Z_2^{(2)}
	$$
where $u+T:=\{u+v:v\in T\}$. Applying \Cref{lemma_sketch_lm1} to $Z_2^{(i)}\, (i=1,2,3)$ we get
\begin{align*}
\P\lp Z_2\leq \frac{\delta}{4} \rp 
&= 1- \P\lp Z_2 >\frac{\delta}{4}\rp \\
&\geq 1- \P\lp Z_2^{(1)}>\frac{\delta}{12}\rp - \P\lp Z_2^{(2)}>\frac{\delta}{12}\rp - \P\lp Z_2^{(3)}>\frac{\delta}{12} \rp \\
&\geq 1-9e^{-cm\delta^2/(K^2\log K)},
\end{align*}
provided $m\geq CK^2\log Kw^2\delta^{-2}$ with
$ w:=\max\{ w(u+T),\, w(\{u\}),\, w(T) \}. $\\
By the properties of Gaussian width, we claim that $w=w(T)$. In fact,
$$ w(\{u\}) =\E \sup_{x\in \{u\}}\inp{g}{x} = \E \inp{g}{u} = 0, $$
$$ w(u+T) = \E \sup_{v\in T}\inp{g}{u+v} = \E \lp \inp{g}{u} + \sup_{v\in T}\inp{g}{v} \rp = \E \inp{g}{u} + \E\sup_{v\in T}\inp{g}{v}=w(T). $$

Combining the bounds for $Z_1$ and $Z_2$ we have
$$ 2\frac{Z_2}{Z_1}\leq 2\frac{\delta/4}{1-\delta/2}\leq \delta $$
with probability at least $1-12e^{-cm\delta^2/(K^2\log K)}$. This completes the proof.
\end{proof}

\subsection{Favorable Landscape for Blind Demodulation with Generative Priors}

In this section, we give a concrete example where the improvement on the sub-Gaussian parameter $K$ can be important through blind demodulation with generative priors.

Blind demodulation aims to recover two signals $x_0, y_0 \in \R^l$ from observation $z_0 = x_0 \circ y_0$, where $\circ$ denotes componentwise multiplication. Due to the inherent nature of ambiguity of the solutions from $z_0$, one usually assume that the signals  come with some structure. A traditional way to model this structure is through a  sparsity prior with respect to a basis such as wavelet or the Discrete Cosine Transform basis in case the signals are images. 

On the other hand, with recent development in deep learning, the generative adversarial network (GAN) is turning out to be very effective in generating realistic synthetic images, which naturally indicates that we may model a certain type of image signals as outputs of GAN. Especially in the inverse problems like compressed sensing, phase retrieval including this blind demodulation, practitioners have observed an order of magnitude sample (observation) complexity improvement over the sparsity prior \cite{bora2017compressed,lucas2018using,hand2017global}.

This alternative model is called the generative prior and as a consequence is becoming a new promising model for modern signal processing \cite{bora2017compressed,hand2017global,hand2018phase,hand2019global}. 
In Hand and Joshi \cite{hand2019global}, the authors provide a global landscape guarantee for blind demodulation problem with generative priors, and they applied our Bernstein's inequality in their proof.

With generative priors, unknown signals $x_0,y_0$ are assumed to be in the range of two generative neural networks $\mathcal{G}^{(1)}$ and $\mathcal{G}^{(2)}$ respectively. More precisely, $\mathcal{G}^{(1)}: \R^n \rightarrow \R^l$ is a $d$-layer network, $\mathcal{G}^{(2)}: \R^p \rightarrow \R^l$ is a $s$-layer network and they can be written as 
\begin{align*}
\mathcal{G}^{(1)}(h) &= \text{relu}\lp W^{(1)}_d \dots \text{relu}\lp W^{(1)}_2 \text{relu}\lp W^{(1)}_1 h\rp \rp \dots \rp, \\
\mathcal{G}^{(2)}(m) &= \text{relu}\lp W^{(2)}_s \dots \text{relu}\lp W^{(2)}_2 \text{relu}\lp W^{(2)}_1 s\rp \rp \dots \rp,
\end{align*}
where $\text{relu}$ is the Rectified Linear Unit activation function given by $\text{relu}(x) = \max\{x, 0\}$ and $W^{(1)}_i, W^{(2)}_j$ for $i \in \{1, \dots, d\}$ and $j \in \{1, \dots, s\}$ are weight matrices. 

The weight matrices are normally obtained in the training process of the networks but the empirical evidence in \cite{arora2015deep} suggests that they behave a ``random-like" quantity . Based on this phenomenon, the authors of \cite{hand2019global}  made the following additional assumptions on the networks $\mathcal{G}^{(1)}$ and $\mathcal{G}^{(2)}$ to facilitate analysis further:
\begin{itemize}
    \setlength\itemsep{0em}
	\item[A1.] The weight matrices are random Gaussian matrices. 
	\item[A2.] The dimension of each layer increases at least logarithmically. 
	\item[A3.] The last layer dimension $l$ satisfies, up to log factors, $l \gtrsim n^2 + p^2$. 
\end{itemize}
The signals can then be recovered by finding their latent codes $h_0\in\R^n$ and $m_0\in\R^p$ such that $x_0 = \mathcal{G}^{(1)}(h_0)$ and $y_0 = \mathcal{G}^{(2)}(m_0)$. This leads to the following empirical risk minimization program:
$$ \min_{h\in\R^n,m\in\R^p} f(h,m):=\frac{1}{2}\| \mathcal{G}^{(1)}(h_0)\circ  \mathcal{G}^{(2)}(m_0) - \mathcal{G}^{(1)}(h)\circ  \mathcal{G}^{(2)}(m)\|_2^2. $$
Note that there is a scaling ambiguity in this problem since it does not distinguish points on curve $\{(ch,\frac{1}{c}m):c>0\}$ for any given $(h,m)$, thus one can only hope to find the solution curve $\{(ch_0,\frac{1}{c}m_0):c>0\}$. The authors in \cite{hand2019global} showed that under assumptions A1-A3, two conditions that are called the Weight Distributed Condition (WDC) and the joint-WDC are met. These conditions guarantee a favorable landscape for the objective function $f(h,m)$, namely $f$ has a descent direction at all points outside of a small neighborhood of four curves containing the solution.

One of the important ingredients in their proof is concentration bounds for singular values of random matrices. When they showed that the joint-WDC condition is satisfied by concentration argument, they were able to improve the requirement in assumption A3 from, up to log factors, $l \gtrsim n^3 + p^3$ to $l \gtrsim n^2 + p^2$. Such improvement is made possible by our new Bernstein's inequality with refined sub-exponential parameter dependence. This $n^2 + p^2$ sample complexity matches the one in the previous recovery guarantees with sparsity prior (in which case $n$ and $p$ denotes the sparsity levels), but potentially better since the latent code dimension is oftentimes smaller than a sparsity level with respect to a particular basis. 

See Theorem 2, Theorem 5, Lemma 8, and Lemma 9 in \cite{hand2019global} for more details.

\section{Conclusion}
\label{sec_conclusion}
In this article, we proved the optimal concentration bound for sub-Gaussian random matrices on sets. Namely, with high probability,
$$ \sup_{x\in T} \left|\frac{1}{\|B\|_F}\|BAx\|_2-\|x\|_2\right| \lesssim \frac{K\sqrt{\log K}}{\sqrt{\mathrm{sr}(B)}} (w(T)+\mathrm{rad}(T)),
    $$
where $B\in \R^{l\times m}$ is an arbitrary matrix, $A\in\R^{m\times n}$ is an (mean zero) isotropic and sub-Gaussian random matrix, $T\in\R^n$ is the set, $K$ is the sub-Gaussian parameter of $A$, $\mathrm{sr}(B)$ is the stable rank of $B$, $w(T)$ is the Gaussian width of $T$ and $\mathrm{rad}(T):=\sup_{y\in T}\|y\|_2$. Compared to the previous work in \cite{liaw2017simple}, this result generalizes by allowing an arbitrary matrix $B$ while improves the dependency on the sub-Gaussian parameter from $K^2$ to the optimal $K\sqrt{\log K}$. Consequently, this can lead to a tighter concentration bound even in the cases where the sub-Gaussian matrix $BA$ have correlated rows.
It is also worth noting that dependence on $w(T)+\mathrm{rad}(T)$ is optimal in general as well.

We also proved, under extra moment conditions, a new Bernstein type inequality and a new Hanson-Wright inequality. The extra conditions here are bounded first absolute moment (e.g. $\E |Y_i|\leq 2$) for Bernstein's inequality and bounded second moment (e.g. $\E X_i^2=1$) for Hanson-Wright inequality. In many cases, these conditions can be easily met -- for example, they are implied by the isotropic condition of random variables or vectors. In general, both of our new inequalities give improved tail bounds in the sub-Gaussian regime, which is the regime of interest in many applications as demonstrated in \Cref{sec_applications}.

\section{Acknowledgements}
Y.~Plan is partially supported by an NSERC Discovery Grant (22R23068), a PIMS CRG 33: High-Dimensional Data Analysis, and a Tier II Canada Research Chair in Data Science. {\"O}.~Y{\i}lmaz is partially supported by an NSERC Discovery Grant (22R82411) and PIMS CRG 33: High-Dimensional Data Analysis. H.~Jeong is funded in part by the University of British Columbia Data Science Institute (UBC DSI) and by the Pacific Institute of Mathematical Sciences (PIMS). 
The authors of this paper would also like to thank Babhru Joshi for the helpful discussions and providing us with an important application in \Cref{sec_applications}.

\bibliography{optimaltail}

\begin{thebibliography}{10}

\bibitem{achlioptas2003database}
D.~Achlioptas.
\newblock Database-friendly random projections: Johnson-lindenstrauss with
  binary coins.
\newblock {\em Journal of computer and System Sciences}, 66(4):671--687, 2003.

\bibitem{adamczak2005logarithmic}
R.~Adamczak.
\newblock Logarithmic sobolev inequalities and concentration of measure for
  convex functions and polynomial chaoses.
\newblock {\em Bulletin of the Polish Academy of Sciences Mathematics},
  53(2):221--238, 2005.

\bibitem{adamczak2015note}
R.~Adamczak.
\newblock A note on the hanson-wright inequality for random vectors with
  dependencies.
\newblock {\em Electronic Communications in Probability}, 20, 2015.

\bibitem{arora2015deep}
S.~Arora, Y.~Liang, and T.~Ma.
\newblock Why are deep nets reversible: A simple theory, with implications for
  training.
\newblock {\em arXiv preprint arXiv:1511.05653}, 2015.

\bibitem{bora2017compressed}
A.~Bora, A.~Jalal, E.~Price, and A.~G. Dimakis.
\newblock Compressed sensing using generative models.
\newblock In {\em Proceedings of the 34th International Conference on Machine
  Learning-Volume 70}, pages 537--546. JMLR. org, 2017.

\bibitem{candes2014mathematics}
E.~J. Cand{\`e}s.
\newblock Mathematics of sparsity (and a few other things).
\newblock In {\em Proceedings of the International Congress of Mathematicians,
  Seoul, South Korea}, volume 123, 2014.

\bibitem{chandar2008negative}
V.~Chandar.
\newblock A negative result concerning explicit matrices with the restricted
  isometry property.
\newblock {\em preprint}, 2008.

\bibitem{cohen2020multi}
A.~Cohen, N.~Shlezinger, A.~Solomon, Y.~C. Eldar, and M.~M{\'e}dard.
\newblock Multi-level group testing with application to one-shot pooled
  covid-19 tests.
\newblock {\em arXiv preprint arXiv:2010.06072}, 2020.

\bibitem{2015dirksen}
S.~Dirksen.
\newblock Tail bounds via generic chaining.
\newblock {\em Electronic Journal of Probability}, 20, 2015.

\bibitem{Dirksen2016dimension}
S.~Dirksen.
\newblock Dimensionality reduction with subgaussian matrices: a unified theory.
\newblock {\em Foundations of Computational Mathematics}, 16(5):1367--1396,
  2016.

\bibitem{foucart2013mathematical}
S.~Foucart and H.~Rauhut.
\newblock {\em A Mathematical Introduction to Compressive Sensing}.
\newblock Birkh\"{a}user, New York, NY, 2013.

\bibitem{gordon1988milman}
Y.~Gordon.
\newblock On milman's inequality and random subspaces which escape through a
  mesh in $\mathbb{R}^n$.
\newblock In {\em Geometric Aspects of Functional Analysis}, pages 84--106.
  Springer, 1988.

\bibitem{hand2019global}
P.~Hand and B.~Joshi.
\newblock Global guarantees for blind demodulation with generative priors.
\newblock In {\em Advances in Neural Information Processing Systems}, pages
  11531--11541, 2019.

\bibitem{hand2018phase}
P.~Hand, O.~Leong, and V.~Voroninski.
\newblock Phase retrieval under a generative prior.
\newblock In {\em Advances in Neural Information Processing Systems}, pages
  9136--9146, 2018.

\bibitem{hand2017global}
P.~Hand and V.~Voroninski.
\newblock Global guarantees for enforcing deep generative priors by empirical
  risk.
\newblock In {\em Proceedings of Machine Learning Research}, volume~75, pages
  1--8, 2018.

\bibitem{hitczenko1998hypercontractivity}
P.~Hitczenko, S.~Kwapie\'{n}, W.~Li, G.~Schechtman, T.~Schlumprecht, and
  J.~Zinn.
\newblock Hypercontractivity and comparison of moments of iterated maxima and
  minima of independent random variables.
\newblock {\em Electronic Journal of Probability}, 3, 1998.

\bibitem{kane2014sparser}
D.~M. Kane and J.~Nelson.
\newblock Sparser johnson-lindenstrauss transforms.
\newblock {\em Journal of the ACM (JACM)}, 61(1):4, 2014.

\bibitem{2005Klartag}
B.~Klartag and S.~Mendelson.
\newblock Empirical processes and random projections.
\newblock {\em Journal of Functional Analysis}, 225(1):229, 2005.

\bibitem{klochkov2018uniform}
Y.~Klochkov and N.~Zhivotovskiy.
\newblock Uniform hanson-wright type concentration inequalities for unbounded
  entries via the entropy method.
\newblock {\em arXiv preprint arXiv:1812.03548}, 2018.

\bibitem{krahmer2015compressive}
F.~Krahmer, D.~Needell, and R.~Ward.
\newblock Compressive sensing with redundant dictionaries and structured
  measurements.
\newblock {\em SIAM Journal on Mathematical Analysis}, 47(6):4606--4629, 2015.

\bibitem{Krahmer2014}
F.~Krahmer and H.~Rauhut.
\newblock Structured random measurements in signal processing.
\newblock {\em GAMM-Mitteilungen}, 37(2):217--238, 2014.

\bibitem{kueng2017robust}
R.~Kueng and P.~Jung.
\newblock Robust nonnegative sparse recovery and the nullspace property of 0/1
  measurements.
\newblock {\em IEEE Transactions on Information Theory}, 64(2):689--703, 2017.

\bibitem{larsen2014johnson}
K.~G. Larsen and J.~Nelson.
\newblock The johnson-lindenstrauss lemma is optimal for linear dimensionality
  reduction.
\newblock In {\em 43rd International Colloquium on Automata, Languages, and
  Programming, {ICALP} 2016, July 11-15, 2016, Rome, Italy}, pages 82:1--82:11,
  2016.

\bibitem{liaw2017simple}
C.~Liaw, A.~Mehrabian, Y.~Plan, and R.~Vershynin.
\newblock A simple tool for bounding the deviation of random matrices on
  geometric sets.
\newblock In {\em Geometric Aspects of Functional Analysis}, pages 277--299.
  Springer, 2017.

\bibitem{lucas2018using}
A.~Lucas, M.~Iliadis, R.~Molina, and A.~K. Katsaggelos.
\newblock Using deep neural networks for inverse problems in imaging: beyond
  analytical methods.
\newblock {\em IEEE Signal Processing Magazine}, 35(1):20--36, 2018.

\bibitem{matouvsek2008variants}
J.~Matou{\v{s}}ek.
\newblock On variants of the johnson--lindenstrauss lemma.
\newblock {\em Random Structures \& Algorithms}, 33(2):142--156, 2008.

\bibitem{2007Mendelson}
S.~Mendelson, A.~Pajor, and N.~Tomczak-Jaegermann.
\newblock Reconstruction and subgaussian operators in asymptotic geometric
  analysis.
\newblock {\em Geometric and Functional Analysis}, 17(4):1248--1282, 2007.

\bibitem{2017oymak}
S.~Oymak and J.~A. Tropp.
\newblock Universality laws for randomized dimension reduction, with
  applications.
\newblock {\em Information and Inference: A Journal of the IMA}, 7(3):337--446,
  2017.

\bibitem{2015pilanci}
M.~Pilanci and M.~J. Wainwright.
\newblock Randomized sketches of convex programs with sharp guarantees.
\newblock {\em IEEE Transactions on Information Theory}, 61(9):5096--5115,
  2015.

\bibitem{robert1990ash}
B.~Robert.
\newblock {\em Ash. Information Theory}.
\newblock Dover Publications Inc., 1990.

\bibitem{rudelson2013hanson}
M.~Rudelson and R.~Vershynin.
\newblock Hanson-wright inequality and sub-gaussian concentration.
\newblock {\em Electronic Communications in Probability}, 18, 2013.

\bibitem{saab2018compressed}
R.~Saab, R.~Wang, and {\"O}.~Y{\i}lmaz.
\newblock From compressed sensing to compressed bit-streams: practical
  encoders, tractable decoders.
\newblock {\em IEEE Transactions on Information Theory}, 64(9):6098--6114,
  2018.

\bibitem{samson2000concentration}
P.-M. Samson.
\newblock {Concentration of measure inequalities for Markov chains and
  $\Phi$-mixing processes}.
\newblock {\em The Annals of Probability}, 28(1):416--461, 2000.

\bibitem{Schechtman2006}
G.~Schechtman.
\newblock Two observations regarding embedding subsets of euclidean spaces in
  normed spaces.
\newblock {\em Advances in Mathematics}, 200(1):125--135, 2006.

\bibitem{shental2020efficient}
N.~Shental, S.~Levy, V.~Wuvshet, S.~Skorniakov, B.~Shalem, A.~Ottolenghi,
  Y.~Greenshpan, R.~Steinberg, A.~Edri, R.~Gillis, et~al.
\newblock Efficient high-throughput sars-cov-2 testing to detect asymptomatic
  carriers.
\newblock {\em Science advances}, 6(37):eabc5961, 2020.

\bibitem{talagrand2014upper}
M.~Talagrand.
\newblock {\em Upper and lower bounds for stochastic processes: modern methods
  and classical problems}, volume~60.
\newblock Springer Science \& Business Media, 2014.

\bibitem{vershynin_2018}
R.~Vershynin.
\newblock {\em High-Dimensional Probability: An Introduction with Applications
  in Data Science}.
\newblock Cambridge Series in Statistical and Probabilistic Mathematics.
  Cambridge University Press, 2018.

\bibitem{vu2015random}
V.~Vu and K.~Wang.
\newblock Random weighted projections, random quadratic forms and random
  eigenvectors.
\newblock {\em Random Structures \& Algorithms}, 47(4):792--821, 2015.

\bibitem{2014woodruff}
D.~P. Woodruff.
\newblock Sketching as a tool for numerical linear algebra.
\newblock {\em Foundations and Trends{\textregistered} in Theoretical Computer
  Science}, 10(1--2):1--157, 2014.

\bibitem{yang2017randomized}
Y.~Yang, M.~Pilanci, M.~J. Wainwright, et~al.
\newblock Randomized sketches for kernels: Fast and optimal nonparametric
  regression.
\newblock {\em The Annals of Statistics}, 45(3):991--1023, 2017.

\end{thebibliography}
\bibliographystyle{abbrv}

\newpage
\appendix

\section{Properties of \texorpdfstring{$\psi_\alpha$}{}-Norm}
\label{appendix_psi_alpha_properties}

\begin{prop}
\label{prop_psi_alpha}
Let $X$, $Y$ be random variables and let $\alpha \geq 1$.
\begin{itemize}
\item[$\mathrm{(a)}$] If $\|X\|_{\psi_\alpha}\leq K<\infty$, then $\P(|X|\geq t)\leq 2\exp(-t^\alpha/K^\alpha)$ for all $t\geq 0$;
\item[$\mathrm{(b)}$] If $\P(|X|\geq t)\leq 2\exp(-t^\alpha/K^\alpha)$ for all $t\geq 0$ and some $K>0$, then $\|X\|_{\psi_\alpha}\leq \sqrt{3}K$;
\item[$\mathrm{(c)}$] $\|X^p\|_{\psi_\alpha}=\|X\|_{\psi_{p\alpha}}^p$ for all $p\geq 1$. In particular, $\|X^2\|_{\psi_1}=\|X\|_{\psi_{2}}^2$;
\item[$\mathrm{(d)}$] $\|XY\|_{\psi_\alpha}\leq \|X\|_{\psi_{p\alpha}}\|Y\|_{\psi_{q\alpha}}$ for $p,q\in (1,\infty)$ such that $\frac{1}{p}+\frac{1}{q}=1$. In particular, $\|XY\|_{\psi_1}\leq \|X\|_{\psi_{2}}\|Y\|_{\psi_{2}}$;
\item[$\mathrm{(e)}$] $\E |X|^p \leq \lp Cp^\frac{1}{\alpha} \|X\|_{\psi_\alpha} \rp ^p$ for all $p\geq 1$ and some absolute constant $C\leq 4$;
\item[$\mathrm{(f)}$] $\|X -\E X\|_{\psi_\alpha} \leq C\|X\|_{\psi_\alpha}$ for some absolute constant $C\leq 7$;
\item[$\mathrm{(g)}$] $\|X\|_{\psi_\alpha}\leq C\|X\|_{\psi_\beta}$ for all $\beta\geq\alpha$ and some absolute constant $C\leq 3$.
\end{itemize}
\end{prop}
In particular, properties (a) and (b) implies that a random variable is sub-Gaussian (or sub-exponential) if and only if its tail probability is bounded by a Gaussian (or exponential) random variable. Properties (c) and (d) tell us if $X$ and $Y$ are both sub-Gaussian, then $X^2$, $Y^2$ and $XY$ are all sub-exponential. Property (e) tells us for $p\geq 1$, all $p$-th moments of $X$ exist whenever $\|X\|_{\psi_\alpha}$ is finite. Property (f) tells us we can always center random variables without changing their $\psi_\alpha$-norm up to a constant factor. Property (g) tells us all sub-Gaussian random variables are also sub-exponential random variables.

\begin{proof}[\bf Proof of \Cref{prop_psi_alpha}]
\hfill\\
\noindent(a) This follows from Markov's inequality and definition of $\psi_\alpha$ norm.
$$
\P (|X|\geq t) = \P\lp e^{X^\alpha/K^\alpha} \geq e^{t^\alpha/K^\alpha} \rp
\leq \exp(-t^\alpha/K^\alpha)\E e^{X^\alpha/K^\alpha}
\leq 2\exp(-t^\alpha/K^\alpha).
	$$

\noindent(b)
With a change of variable $s=e^\frac{u}{3K^\alpha}$ on interval $(1,\infty)$ we have
\[
\E \exp(\frac{|X|^\alpha}{3K^2}) = \int_{0}^\infty \P \lp e^\frac{|X|^\alpha}{3K^\alpha}\geq s \rp ds
= \int_{0}^1 1\, ds + \frac{1}{3K^\alpha} \int_0^\infty \P(|X|^\alpha\geq u)e^\frac{u}{3K^\alpha} du.
\]
Since $\P(|X|^\alpha\geq u)=\P(|X|\geq u^{1/\alpha})\leq 2\exp(-u/K^2)$, we get
\[
\E \exp(\frac{|X|^\alpha}{3K^2}) \leq 1 + \frac{2}{3K^\alpha} \int_0^\infty \exp(-\frac{u}{K^\alpha}+\frac{u}{3K^\alpha}) du = 2.
\]

\noindent(c) 
This follows from definition.
\begin{align*}
\|X^p\|_{\psi_\alpha}
&= \inf\{ t>0: \E \exp(|X|^{p \alpha}/t^{\alpha}) \leq 2\} \\
&= \inf\{ u^p: u>0 \,\text{ and }\, \E \exp(|X|^{p \alpha}/u^{p\alpha}) \leq 2\} \\
&= \lp \, \inf\{ u>0: \E \exp(|X|^{p \alpha}/u^{p\alpha}) \leq 2\} \, \rp ^p \\
&= \|X\|_{\psi_{p\alpha}}^p.
\end{align*}

\noindent(d) 
Without loss of generality, we can assume $\|X\|_{\psi_{p\alpha}}=\|Y\|_{\psi_{q\alpha}}=1$.\\
By Young's inequality $|ab|\leq \frac{1}{p}|a|^p +\frac{1}{q}|b|^q$ we have
$$
\exp(|XY|^\alpha) 
\leq \exp(\frac{1}{p}|X|^{p\alpha}+\frac{1}{q}|Y|^{q\alpha})
= \exp(\frac{1}{p}|X|^{p\alpha})\exp(\frac{1}{q}|Y|^{q\alpha}).
    $$
Applying Young's inequality again we have
$$
\exp(\frac{1}{p}|X|^{p\alpha})\exp(\frac{1}{q}|Y|^{q\alpha})
\leq \frac{1}{p}\exp(|X|^{p\alpha})+\frac{1}{q}\exp(|Y|^{q\alpha}).
    $$
Therefore
$$
\E \exp(|XY|^\alpha) \leq \frac{1}{p}\E \exp(|X|^{p\alpha})+\frac{1}{q}\E \exp(|Y|^{q\alpha})
\leq \frac{2}{p} + \frac{2}{q} =2.
    $$
This shows $\|XY\|_{\psi_\alpha}\leq 1$.

\vskip .1in
\noindent(e) 
Without loss of generality, we can assume $\|X\|_{\psi_{\alpha}}=1$. Then by property (a), $$\P(|X|\geq t)\leq 2\exp(-t^\alpha)  \;\text{ for } t\geq 0. $$
With a change of variable $u=t^\alpha$ we have
\begin{align*}
\E |X|^p & =\int_0^\infty pt^{p-1}\P(|X|\geq t)\, dt \\
&\leq \int_0^\infty pu^\frac{p-1}{\alpha}2e^{-u}\frac{1}{\alpha}u^{\frac{1}{\alpha}-1} \, du \\
& = \int_0^\infty \frac{2p}{\alpha}u^{\frac{p}{\alpha}-1}e^{-u} \, du \\
&= \frac{2p}{\alpha}\Gamma \lp \frac{p}{\alpha} \rp \\
&= 2\Gamma \lp \frac{p}{\alpha} +1 \rp
\end{align*}
where $\Gamma(\cdot)$ denotes the Gamma function. Note that for $s>0$,
$$
\Gamma(s+1)=\int_0^\infty \lp x^s e^{-\frac{x}{2}} \rp e^{-\frac{x}{2}} dx \leq (2s)^s e^{-s}\int_0^\infty e^{-\frac{x}{2}} dx = 2 \lp \frac{2s}{e} \rp ^s,
    $$
where we used the fact that $x^se^{-\frac{x}{2}}$ attains maximum at $x=2s$ because
$$ \frac{d}{dx} \lp x^se^{-\frac{x}{2}} \rp = x^{s-1} e^{-\frac{x}{2}} \lp s-\frac{x}{2} \rp .$$
Therefore
$$
\E |X|^p \leq 4\lp \frac{2p}{\alpha e} \rp ^\frac{p}{\alpha} =4\lp \frac{2}{\alpha e} \rp ^\frac{p}{\alpha} p^\frac{p}{\alpha} \leq 4 p^\frac{p}{\alpha} \leq \lp 4p^\frac{1}{\alpha} \rp ^p.
    $$

\vskip .1in
\noindent(f) By triangle inequality,
$$ \|X -\E X\|_{\psi_\alpha} \leq \|X\|_{\psi_\alpha} + \|\E X\|_{\psi_\alpha}. $$
Using property (d) and the fact that $\|1\|_{\psi_\alpha}=\lp \frac{1}{\log 2} \rp ^{1/\alpha}$ we have
$$ \|\E X\|_{\psi_\alpha}= |\E X|\cdot \|1\|_{\psi_\alpha} \leq 4 \|X\|_{\psi_\alpha}\cdot  \frac{1}{\log 2}. $$
This completes the proof with $C=1+\frac{4}{\log 4}\approx 6.77$.

\vskip .1in
\noindent(g)
Without loss of generality, we can assume $\|X\|_{\psi_{\beta}}=1$. Then by property (a), 
$$\P(|X|\geq t)\leq 2\exp(-t^\beta) \;\text{ for } t\geq 0.$$
Next we show $$\P(|X|\geq t)\leq 2\cdot 2^{-t^\alpha}=2\exp(-t^\alpha \log 2).$$ 
In fact, this is trivial when $t\in [0,1]$ since $2\cdot 2^{-t^\alpha}\geq 1$.\\
When $t>1$, from $t^{\beta-\alpha}\geq 1>\log 2$ we get
$$
t^\beta \geq t^\alpha \log 2 \;\Rightarrow\; \exp(-t^\beta)\leq \exp(-t^\alpha \log 2)
\;\Rightarrow\; \P(|X|\geq t)\leq 2\exp(-t^\alpha \log 2).
    $$
Therefore by property (b) we have
$$
\|X\|_{\psi_\alpha}\leq \sqrt{3} \lp \frac{1}{\log 2} \rp ^\frac{1}{\alpha} \leq \frac{\sqrt{3}}{\log 2} \approx 2.50.
    $$

\end{proof}

\section{Dependence on Sub-Gaussian Parameter for JL Lemma}
\label{appendix_jlopt}
Here we give an example to demonstrate the $K^2\log K$ dependence in sample complexity for JL Lemma (\Cref{lemma_JL_xy}) is optimal for small $\varepsilon$ and $\delta$. Similar to \Cref{prop_tightness_ex}, this example is also based on scaled Bernoulli distribution. However, this result is not implied by \Cref{prop_tightness_ex} because the latter already assumed $m\geq K^2\log K$.

\begin{prop}
Let $p \in (0,\frac{1}{4})$ and denote $K$ the positive number such that $p^{-1}=K^2\log K$. Let $A\in \R^{m\times n}$ be a random matrix with symmetric i.i.d. entries $A_{ij}$ such that $A_{ij}^2\sim \frac{1}{mp} \cdot \mathrm{\mathbf{Bernoulli}}\lp p \rp$. 
Denote $e_1=(1,0,\dots,0)^T\in \R^n$ and suppose there exist some $\varepsilon,\delta\in (0,\frac{1}{5})$ and an integer $m_0$ such that, as $m$ increases, the probability bound
$$ \P( |\|Ae_1\|_2-1| \geq \varepsilon) \leq \delta
	$$
holds for all $m\geq m_0$. Then
\begin{itemize}
\item[$\mathrm{(a)}$] $\sqrt{m}A$ is isotropic and sub-Gaussian, with sub-Gaussian parameter being no more than $K$.
\item[$\mathrm{(b)}$] $m_0\geq \frac{1}{2}K^2\log K$.
\end{itemize}
\end{prop}
\begin{proof} Part (a) is straightforward to verify (recall $\sqrt{m}\,\|A_{ij}\|_{\psi_2}\leq K$, as shown in \Cref{prop_tightness_ex}). 
For part (b), first notice that
$$
 \P \lp |\|Ae_1\|_2-1| \geq \varepsilon \rp \;\geq\;  \P \lp \|Ae_1\|_2^2 \geq (1+\varepsilon)^2 \rp  \;\geq\;  \P \lp \|Ae_1\|_2^2\geq 1+3\varepsilon \rp,
	$$
and $\|Ae_1\|_2^2\sim \frac{1}{mp}Z_m$ where $Z_m \sim \mathbf{Binomial}(m,p)$. We will prove the result by contradiction. If $m_0p<1/2$, choose $m_1$ to be the largest integer such that $m_1p\leq 1/2$. Then $m_1p\geq 1/2-p \geq 1/4$ and by taking $m=m_1\geq m_0$ we have
$$ \delta \;\geq\; \P(\|Ae_1\|_2^2\geq 1+3\varepsilon) \;=\; \P(Z_{m_1} \geq m_1p(1+3\varepsilon)) 	\;=\; \P(Z_{m_1} \geq 1) \;=\; 1-\P(Z_{m_1} =0). $$
Then
$$ \frac{1}{5}>\delta \geq 1- (1-p)^{m_1} \geq 1 -e^{-m_1p} \geq 1-e^{-\frac{1}{4}} \approx 0.22
$$
where we used inequality $(1-p) \leq e^{-p}$. This completes the proof. 
\end{proof}

\section{A Few Inequalities}
\label{appendix_b}
Here we list and prove the non-standard inequalities used in our proofs.
\begin{itemize}
\item[(a)] $e^x \leq x + \cosh(2x)$ for $x\in \R$.
\item[(b)] $(1-x)^{-\frac{1}{2}}\leq e^x$ for $x\in [0,\frac{1}{2})$.
\item[(c)] $\min\{ 1,\alpha e^{-x}\} \leq 2\exp(-\frac{x\log 2}{\log \alpha}) = 2\exp(-\frac{x}{\log_2 \alpha})$ for $\alpha \geq 2$ and $x\in \R$.
\item[(d)] $(1-x)\lp \frac{2}{x(1-x)} \rp ^ {x^2/2} \leq 1$ for $x\in (0,1)$.
\end{itemize}

\begin{proof} 
(a) From $e^{-2x}\geq 1-2x$ we have
\begin{align*}
2\cosh(2x)+2x-2e^x &= e^{-2x}+ e^{2x} - 2e^x  +2x\\
&\geq 1-2x + (e^x-1)^2-1 +2x \\
&=(e^x-1)^2 \geq 0
\end{align*}

\noindent (b) It suffices to show $f(x):=e^{2x}(1-x)\geq 1$ on $[0,\frac{1}{2})$. \\
Since $f'(x)=e^{2x}(1-2x)>0$ for $x<\frac{1}{2}$, we have $f(x)\geq f(0)=1$ when $x\in [0,\frac{1}{2})$.

\vskip .1in
\noindent (c) When $x\leq \log \alpha$, we have $2\exp(-\frac{x\log 2}{\log \alpha})\geq 2\exp(-\log 2) =1$.\\
When $x>\log \alpha$, notice that
$$\log \alpha -x < \frac{\log 2}{\log \alpha}(\log \alpha -x)=\log 2-\frac{x\log 2}{\log \alpha}.$$
Taking exponential we get $\alpha e^{-x}<2\exp(-\frac{x\log 2}{\log \alpha})$.

\vskip .1in
\noindent (d) It suffices to show $f(x)\leq 0$ where
\[
f(x):= \log(1-x) + \frac{x^2}{2}\log\frac{2}{x(1-x)}, \quad  x\in (0,1).
\]
Taking its derivative we have
\[
f'(x)=-\frac{2(x^2-x)\log\frac{2}{x-x^2} + x-2x^2+2 }{2(1-x)}=-\frac{g(x-x^2)+1-x^2}{2(1-x)}
\]
where $g(t):=-2t\log\frac{2}{t}+t+1$ with $t=x-x^2\in (0,\frac{1}{4}]$.

Since $g'(t)=3+2\log\frac{t}{2}\leq g'(\frac{1}{4})<0$, we know $g$ is decreasing with $g(t)\geq g(\frac{1}{4})>0$. Thus $f'<0$ and $f(x)\leq \lim_{x\to 0}f(x)=0$.
\end{proof}

\end{document}